\documentclass[sigconf]{aamas} 

\usepackage{balance} 
\usepackage{bm}
\usepackage{dsfont}
\usepackage{tikz}
\usepackage{oplotsymbl}
\usepackage{marvosym}
\usepackage{savesym}
\savesymbol{Cross}
\usepackage{bbding}
\restoresymbol{bb}{Cross}
\usetikzlibrary{trees}
\usetikzlibrary{quotes}
\usetikzlibrary{positioning}
\usetikzlibrary{shapes}
\usepackage[font=small,labelfont=bf]{caption}
\usepackage{subcaption}
\usepackage{xspace}
\usepackage{float}
\usepackage{color}
\usepackage{listings}
\usepackage{setspace}
\usepackage{algorithmicx}
\usepackage[Algorithm,ruled]{algorithm}
\usepackage{algpseudocode}
\usepackage{etoolbox}
\usepackage{enumitem}

\usepackage{xr}

\newcommand{\term}[1]{\textbf{\textit{#1}}}
\newcommand{\T}{\textsc{true}}
\newcommand{\F}{\textsc{false}}
\newcommand{\cA}{\mathcal{A}}
\newcommand{\Udist}{\mathbb{U}}
\newcommand{\cI}{\mathcal{I}}
\newcommand{\cV}{\mathcal{V}}
\newcommand{\cW}{\mathcal{W}}
\newcommand{\cO}{\mathcal{O}}
\newcommand{\cT}{\mathcal{T}}
\newcommand{\cR}{\mathcal{R}}
\newcommand{\cS}{\mathcal{S}}
\newcommand{\bsigma}{\bm{\sigma}}
\newcommand{\targprof}{\bm{\Tilde{\sigma}^*}}

\newcommand{\regret}{\mathsf{Reg}}

\newcommand{\numItems}{m}
\newcommand{\thresh}{\nu}
\newcommand{\offer}{\omega}
\newcommand{\infomap}{\phi}
\newcommand{\piBR}{\pi^*}

\newcommand{\given}{\mid}
\newcommand{\gain}{\Gamma}

\providecommand{\abs}[1]{\lvert#1\rvert}

\newcommand{\gengoof}{\textsc{GenGoof}\xspace}
\newcommand{\gengoofK}[1]{\textsc{GenGoof}$_{#1}$\xspace}
\newcommand{\barg}{\textsc{Bargain}\xspace}

\definecolor{player1}{RGB}{213, 94, 0}
\definecolor{player1_1}{RGB}{100, 143, 255}
\definecolor{player1_2}{RGB}{0, 153, 255}
\definecolor{player1_3}{RGB}{226, 146, 57}
\definecolor{player1_4}{RGB}{161, 191, 1}
\definecolor{player1_5}{RGB}{121, 71, 186}
\definecolor{player1_6}{RGB}{193, 32, 224}
\definecolor{player1_7}{RGB}{122, 4, 4}
\definecolor{player1_8}{RGB}{37, 124, 163}
\definecolor{player1_9}{RGB}{210, 153, 94}
\definecolor{player1_10}{RGB}{94, 210, 153}
\definecolor{player1_11}{RGB}{107, 65, 147}
\definecolor{player1_12}{RGB}{184, 126, 122}
\definecolor{player1_13}{RGB}{176, 203, 166}
\definecolor{player1_14}{RGB}{143, 175, 107}
\definecolor{player1_15}{RGB}{232, 244, 218}
\definecolor{player1_16}{RGB}{145, 76, 61}
\definecolor{player1_17}{RGB}{193, 94, 194}
\definecolor{player1_18}{RGB}{153, 0, 0}
\definecolor{player1_19}{RGB}{0, 153, 153}
\definecolor{player1_20}{RGB}{48, 233, 186}

\definecolor{player2}{RGB}{0, 114, 189}
\definecolor{player2_1}{RGB}{135, 100, 255}
\definecolor{player2_2}{RGB}{176,224,230}
\definecolor{player2_3}{RGB}{128, 0, 128}
\definecolor{player2_4}{RGB}{255, 178, 255}
\definecolor{player2_5}{RGB}{205,188,198}
\definecolor{player2_6}{RGB}{214, 198, 164}
\definecolor{player2_7}{RGB}{204, 0, 204}
\definecolor{player2_8}{RGB}{204, 204, 0}
\definecolor{player2_9}{RGB}{184, 173, 241}
\definecolor{player2_10}{RGB}{64, 224, 208}
\definecolor{player2_11}{RGB}{204, 255, 0}
\definecolor{player2_12}{RGB}{229, 193, 0}
\definecolor{chance}{RGB}{0, 158, 115}
\definecolor{otherplayer1}{RGB}{204, 121, 167}
\definecolor{otherplayer2}{RGB}{240, 228, 66}
\definecolor{terminal}{RGB}{253, 236, 160}
\definecolor{oldplayer}{RGB}{150, 150, 150}

\definecolor{gg_player1}{RGB}{128, 0, 128}
\definecolor{gg_player1_1}{RGB}{210, 151, 139}
\definecolor{gg_player1_2}{RGB}{192, 222, 150}
\definecolor{gg_player1_3}{RGB}{0, 105, 150}
\definecolor{gg_player1_4}{RGB}{255, 105, 180}
\definecolor{gg_player1_5}{RGB}{142, 59, 0}
\definecolor{gg_player1_6}{RGB}{152, 152, 255}
\definecolor{gg_player1_7}{RGB}{243, 222, 138}
\definecolor{gg_player1_8}{RGB}{1, 137, 223}
\definecolor{gg_player1_9}{RGB}{214, 85, 85}
\definecolor{gg_player1_10}{RGB}{218, 165, 32}
\definecolor{gg_player1_11}{RGB}{80, 200, 120}
\definecolor{gg_player1_12}{RGB}{120, 80, 200}
\definecolor{gg_player1_13}{RGB}{204,161,201}
\definecolor{gg_player1_14}{RGB}{190,209,227}
\definecolor{gg_player1_15}{RGB}{238,251,29}

\definecolor{gg_player2}{RGB}{230, 0, 230}
\definecolor{gg_player2_1}{RGB}{65, 105, 225}
\definecolor{gg_player2_2}{RGB}{68, 170, 153}
\definecolor{gg_player2_3}{RGB}{230, 97, 0}

\definecolor{gg_chance}{RGB}{126, 127, 154}

\setcopyright{ifaamas}
\acmConference[AAMAS '25]{Proc.\@ of the 24th International Conference
on Autonomous Agents and Multiagent Systems (AAMAS 2025)}{May 19 -- 23, 2025}
{Detroit, Michigan, USA}{Y.~Vorobeychik, S.~Das, A.~Nowe (eds.)}
\copyrightyear{2025}
\acmYear{2025}
\acmDOI{}
\acmPrice{}
\acmISBN{}

\acmSubmissionID{917}

\title[Policy Abstraction and Nash Refinement in TE-PSRO]{Policy Abstraction and Nash Refinement in Tree-Exploiting PSRO}

\author{Christine Konicki$^*$}\thanks{$^*$Konicki worked on this paper while a PhD student at the University of Michigan}
\affiliation{
  \institution{Michigan Tech Research Institute}
  \city{Ann Arbor}
  \country{USA}}
\email{ckonicki@mtu.edu}

\author{Mithun Chakraborty}
\affiliation{
  \institution{University of Michigan}
  \city{Ann Arbor}
  \country{USA}}
\email{dcsmc@umich.edu}

\author{Michael P. Wellman}
\affiliation{
  \institution{University of Michigan}
  \city{Ann Arbor}
  \country{USA}}
\email{wellman@umich.edu}

\begin{abstract}
Policy Space Response Oracles (PSRO) interleaves empirical game-theoretic analysis with deep reinforcement learning (DRL) to solve games too complex for traditional analytic methods. 
Tree-exploiting PSRO (TE-PSRO) is a variant of this approach that iteratively builds a coarsened empirical game model \textit{in extensive form} using data obtained from querying a simulator that represents a detailed description of the game.
We make two main methodological advances to TE-PSRO that enhance its applicability to complex games of imperfect information. 
First, we introduce a scalable representation for the empirical game tree where edges correspond to \textit{implicit policies} learned through DRL.
These policies cover conditions in the underlying game abstracted in the game model, supporting sustainable growth of the tree over epochs. 
Second, we leverage extensive form in the empirical model by employing refined Nash equilibria to direct strategy exploration.
To enable this, we give a modular and scalable algorithm based on generalized backward induction for computing a subgame perfect equilibrium (SPE) in an imperfect-information game. 
We experimentally evaluate our approach on a suite of games including an alternating-offer bargaining game with outside offers; our results demonstrate that TE-PSRO converges toward equilibrium faster when new strategies are generated based on SPE rather than Nash equilibrium, and with reasonable time/memory requirements for the growing empirical model.
\end{abstract}


\newcommand{\BibTeX}{\rm B\kern-.05em{\sc i\kern-.025em b}\kern-.08em\TeX}

\begin{document}


\pagestyle{fancy}
\fancyhead{}


\maketitle 


\section{Introduction}\label{sec:intro}

\term{Empirical game-theoretic analysis} (EGTA) \citep{Tuyls20,Wellman16} reasons about complex game scenarios through \term{empirical game} models estimated from simulation data.
A popular form of EGTA is the \term{policy space response oracles} (PSRO) framework \citep{psro17} (Fig.~\ref{fig:psro}), in which the empirical game is iteratively extended by adding best responses (BRs) derived from \term{deep reinforcement learning} (DRL).
The vast majority of prior work on EGTA and PSRO \citep{Bighashdel24,Wellman24tg} represents the empirical game in normal form even though the real underlying game consists of agents' strategies interacting via sequential decisions under various unknowns. 
\citet{McAleer21} introduced \term{XDO} as an alternative to PSRO that maintains an empirical game in extensive form, to capture a combinatorial space of strategies with the choice of actions at each decision point in the game tree. 
We originally proposed and evaluated a
\term{tree-exploiting} version of EGTA (TE-EGTA) that maintains an empirical game in an extensive form based on a coarsening of the underlying game \citep{konicki22}. 
This work demonstrated that a significant improvement in model accuracy and strategy exploration, compared to normal-form EGTA, can be achieved by using the tree structure to model even a little of the information-revelation and action-conditioning patterns of the underlying game.

\begin{figure}[!htb]
    \centering
    \includegraphics[width=\columnwidth]{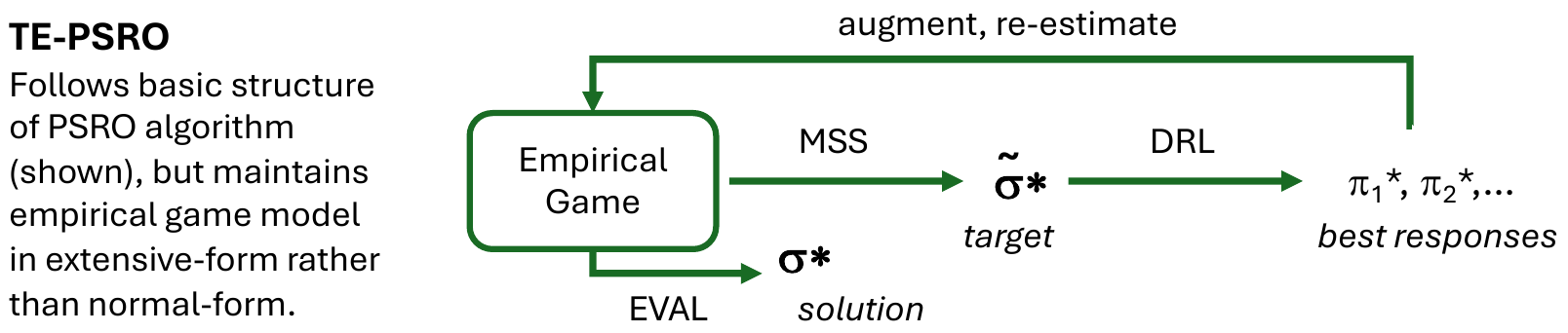}
    \caption{Basic PSRO loop.
    In each iteration (or \term{epoch)}, an empirical game model is extended, based on best responses (BRs) to target profile $\targprof$ derived from the current empirical game by the solver MSS.
    The BRs are computed using deep RL applied to the game simulator. EVAL is a solver (not necessarily the MSS) applied to a model to assess its quality.
    }
    \label{fig:psro}
\end{figure}

A key step in PSRO is augmenting the empirical game model with new BR results.
This is straightforward for a normal-form model: add the new strategies to the game matrix, and estimate payoffs for new profiles using simulation.
For PSRO using an extensive-form model, where we can no longer treat a BR as an atomic entity, we face new  questions such as the following. In what precise sense does the empirical game tree \textit{coarsen} or \textit{abstract} the underlying multiagent scenario? Relatedly, how should we incorporate elements of the BRs (detailed policy specifications) at appropriate places in the empirical game tree?
The way \citet{konicki22} address these challenges in their approach towards tree-exploiting PSRO (TE-PSRO) is by systematically coarsening away select (non-strategic) stochastic events in the underlying game. This approach is not general or scalable in the sense that it relies on the use of stochastic events to model imperfect information in the underlying game.
In this paper, we reformulate TE-PSRO by developing a method to abstract broad swaths of the \textit{state} and \textit{observation} spaces of the underlying game, providing a more implicit rendering of games with complex information structure, including high degrees of imperfect information. We also introduce other methodological advances that enhance the power of TE-PSRO in many directions.

First, to address the abstraction issue, we use two distinct formulations for the game of interest.
The \textit{underlying game as represented by the simulator} is defined in terms of a state space, agent actions, and the observations, successor states, and rewards (stochastically) resulting from applying actions in a given state.
This formulation is the natural one for DRL algorithms, which can interact directly with the simulator.
On the other hand, the \textit{empirical game model} employs an extensive-form tree representation, which is the natural object of game-solving algorithms.
To bridge the two formulations, the 
edges in the empirical game tree correspond to \term{abstract policies} executable in the simulator.
These abstract policies, derived through DRL, map agent observations to actions.%
\footnote{\citet{McAleer21} define a variant of XDO called Neural XDO that likewise employs policies represented as neural networks.
Rather than incorporate these as abstract policies in an empirical game tree, Neural XDO instead relies on methods like neural fictitious self-play \citep{heinrich2016deep} that perform game analysis directly in the space of neural network policies.}
The empirical game includes only select elements of this history, rendering much of the simulator state space and observation space implicit in this abstracted formulation.
At what level to incorporate observation details in the game model is a design choice, entailing tradeoffs in computation and fidelity.

Second, we face additional computational tradeoffs regarding how much to elaborate the game model based on DRL computations. In each iteration (or \textit{epoch}), PSRO solves the current empirical game using a \term{meta-strategy solver} (MSS).
It then derives an approximate best response for each player using DRL, assuming the other players follow the latest MSS solution. 
A straightforward approach would apply the derived response throughout the empirical game \citep{konicki22,McAleer21}.
This, however, could lead the game tree to grow at an exponential rate, severely limiting the feasible number of PSRO epochs.
We propose to control this growth by adding the new DRL policies to only a select set of information sets in the empirical game. 

A final question we address regards the choice of MSS for TE-PSRO.
Previous research has considered a range of MSSs that operate on normal-form games, and observed a significant impact on PSRO efficacy \citep{psro17,Wang22mw}. We now have the opportunity to consider new MSSs that exploit the tree structure of an empirical game in extensive form, in particular, refined Nash equilibria.

To illustrate and evaluate our upgraded TE-PSRO approach, we employ two (perfect-recall) games with significantly different imperfect information structures that non-trivially extend stylized games from the literature. In our first game that we call \barg, two players with private valuations over a set of indivisible items negotiate how to split the set up between them through an alternating-offer protocol.
The scenario features imperfect information about the other party's valuation as well as choices for signaling regarding the value of a private \textit{outside option} that each player has recourse to in the event of negotiation failure. This sequential bargaining game extends a well-known two-party multi-issue negotiation task \citep{devault_dond15,lewis_dond17,fatima2014principles,li_dond23}. 

The second game is a general-sum abstraction of the card game Goofspiel \cite{fixx_goof72} that we call \gengoof. As a clean model of multi-round multiagent interactions with considerable strategic depth, Goofspiel has been extensively analyzed in the game-theory literature, more recently serving as a testbed for game-playing AI algorithms \citep{bosansky_sim16}. \gengoof proceeds over an arbitrary number of rounds, each including a discrete stochastic event (defined over a support diminishing every round by the single realized outcome) followed by all players choosing one discrete action each (effectively simultaneously); the payoff of each player at game termination is the sum of arbitrary per-round rewards. 

Our key contributions are:
\begin{itemize}
    \item A general scheme for abstract policy representation that supports flexible implementation of TE-PSRO for complex games of imperfect information (\S\ref{sec:abstraction} and~\S\ref{sec:best_response}).
    \item An approach to control the growth of empirical game trees through selective incorporation of best responses at particular information sets (\S\ref{sec:best_response_augment} and App.~\ref{sec:expand_game_tree}).
    \item A new algorithm that computes a subgame perfect equilibrium (SPE) of an imperfect-information game (\S\ref{sec:SPE}).
    \item Experimental demonstration of the efficacy of SPE over NE as MSS for TE-PSRO on a variety of complex sequential games of imperfect information (\S\ref{sec:game_description} and \S\ref{sec:experiments}); our experiments address three aspects of the complete TE-PSRO loop: the effectiveness of our augmentation heuristic in controlling the empirical game growth rate, the power of our SPE computation algorithm, and a comparison of MSS choices including refined equilibria which are feasible only for extensive-form empirical games. 
\end{itemize}
The code base used for our experiments is available at \url{https://github.com/ckonicki-umich/AAMAS25}.


\section{Technical Preliminaries}\label{sec:prelim}

An \term{extensive-form game} (EFG) is a tuple 
\[G=~\langle~N,~H,~V,~\{\mathcal{I}_j\}_{j=0}^n,~\{A_j\}_{j=1}^n,~X,~P,~u~\rangle,\] 
 where $N = \{0, \dotsc, n\}$ is the set of players or agents; $H$ is a finite tree of histories divided into subsets of \term{terminal nodes} or leaves $Z$ and decision nodes $D$; $V$ is a function assigning each decision node $h$ to an acting player; $\cA_j(\cdot)$ is the set of actions available at each decision node; $u$ is a function mapping each $z \in Z$ to a \term{utility} vector $\{u_j(z)\}_{j=1}^n$;
in games of imperfect information, the set $\cI_j$ is a partition of $V^{-1}(j)$ where each $I \in \cI_j$ is an \term{information set} (or \term{infoset}) of $j$. All nodes $h \in I$ are indistinguishable to player~$j$, meaning their action spaces are also indistinguishable and denoted $A_j(I)$. 
The directed edge connecting any $h \in I$ to its child represents a transition resulting from $V(h)$'s move. We assume 
\term{perfect recall} \citep[Definition 5.2.3]{shoham2008multiagent}. 
A node~$h$ where $V(h) = 0$ is called a \term{chance node} controlled by Nature, with a set of possible outcomes $X(h)$ and probability distribution $P(\cdot \mid h)$ over $X(h)$.

Since the underlying or \term{true game} corresponding to the simulator is too large to be represented directly with a tree, we instead express it in a state-action formulation.
A play of the game is a sequence of actions taken by the players (including Nature), where each action leads to a \term{world state} $w \in \cW$. 
The joint space of actions is given by $\cA = \bigotimes_{j = 1}^n \cA_j$, and the set of legal actions for agent $j$ at world state~$w$ is given by $\cA_j(w) \subseteq \cA_j$.
The probability distribution of the world state $w'$ following joint action $a = (a_1, \dotsc, a_n) \in \cA$ taken in world state $w$ is given by a transition function $\cT(w, a) \in \Delta^{\cW}$. 
Upon transitioning to $w'$ from $w$ via $a$, agent $j$ makes a partial \term{observation} $o_j = \cO_j(w, a, w')$ instead of fully observing $w'$. A reward $\cR_j(w)$ is given to agent $j$ at each $w \in \cW$, and the game ends when a \term{terminal world state} is reached. 
In this formulation, a history at time~$t$ is the sequence of world states and actions given by $h = (w^1, a^1, \dotsc, w^t)$; histories $z \in Z$ where $w^t$ is a terminal world state are terminal histories.
It follows that $R_j(h)$ and $\cA_j(h)$ are the reward and action space for agent~$j$ in the last world state of history~$h$.
An \term{information state} (or \term{infostate}) for agent~$j$, denoted $I_j$, is a sequence of agent~$j$'s observations and actions up to a point~$t$ in the game, given by $I_j(h) = (a^1_j, o^2_j, a^2_j, \dotsc, o^t_j) \equiv I_j$. 
Since agent~$j$ cannot distinguish between the histories of $I_j(h)$, it follows that $\cA_j(I_j(h)) = \cA_j(h)$. The complete set of infostates for player~$j$ is again given by $\cI_j$. 
We use hatted symbols to denote components of the empirical game tree (e.g., $\hat{\cI_j}$ for infosets) to distinguish them from analogous components of the true game.

A \term{pure strategy} for player~$j$ specifies the action that $j$ selects at each information set.
A \term{mixed strategy} $\sigma_j$ defines a probability distribution over the action space at each of $j$'s information sets.
A \term{strategy profile} is given by $\bsigma = (\sigma_1, \dotsc, \sigma_n)$, and $\bsigma_{-j}$ denotes the strategies of all players other than $j$. 
$\Sigma_j$ is the set of all strategies available to player~$j$, and $\Sigma = \bigtimes_{j = 1}^n \Sigma_j$ denotes the space of joint strategy profiles. A terminal history $z$ is reached by $\bsigma$ with a \term{reach probability} $r(z, \bsigma) = \prod_{j \in N} r_j(z, \sigma_j)$ where $r_j(z, \sigma_j)$ is the probability that player~$j$ chooses actions that lead to $z$, including Nature's contribution $r_0(z)$. 
The \term{payoff} of $\bsigma$ to player~$j$ is given by its expected utility $U_j(\bsigma)$. 
Player~$j$'s \term{regret} from playing~$\sigma_j$ as part of~$\bsigma$ is given by
$\regret_j(\bsigma) = \max_{\sigma \in \Sigma_j} U_j(\sigma, \bsigma_{-j}) - U_j(\bsigma)$. 
The profile regret of $\bsigma$ is the sum of player regrets: $\regret(\bsigma) = \sum_{j=1}^n \regret_j(\bsigma)$. 
A strategy profile $\bsigma$ with $\regret(\bsigma) = 0$ is a \term{Nash equilibrium} (NE).

\section{Description of Games Studied}\label{sec:game_description}
We will now describe in detail the two games used in our experimental assessment of TE-PSRO in \S\ref{sec:experiments}. The game analyst using TE-PSRO has no direct access to a game description at this level of detail, but can query a simulator based on such a description for data samples relevant to game histories induced by input strategy profiles.

\subsection{\barg}\label{sec:DOND}

In this game, two players negotiate the division of $\numItems$ discrete items of $\tau$ 
types.
We represent the item pool by a vector $\mathbf{p}$ where the $i^\mathrm{th}$ entry $p_i$, $i \in \{1,\dotsc,\tau\} \equiv [\tau]$, is the number of items of type~$i$; $\sum_{i=1}^\tau p_i = \numItems$. 
Each player~$j \in \{1,2\}$ has a private valuation over the items given by a vector $\mathbf{v}_j$ of non-negative integers such that the $i^\mathrm{th}$ entry $v_{j,i}$ is player~$j$'s value for one item of type~$i$. 
In each game instance, $(\mathbf{v}_1, \mathbf{v}_2)$ are sampled uniformly at random from the collection $\cV$ of all vector pairs satisfying three constraints. 
First, for both players, the total value of all items is a constant: $\mathbf{v}_j \cdot \mathbf{p} = \bar{V}$, $j\in\{1,2\}$. 
Second, each item type must have nonzero value for at least one player: $\forall i \in [\tau] . v_{1, i} + v_{2, i} > 0$. 
Finally, some item type must have nonzero value for both players: $\exists i \in [\tau] . v_{1, i} v_{2, i} > 0$.


An additional feature of our game is that each player~$j$ has a private \term{outside offer} in the form of a vector of items $\mathbf{o}_j$,
defining the fallback payoff the player obtains if no deal is reached.
This offer is drawn from a distribution $P_j(\cdot)$ at the start of each game instance. During negotiation, a player~$j$ may choose to reveal coarsened information about its outside offer to the other player in the form of a binary signal which is $L$ (resp.~$H$) if the value of the offer $\mathbf{o}_j \cdot \mathbf{v}_j$ is at most (resp.\ greater than) a fixed threshold~$\thresh$ where $1 < \thresh < \bar{V}$.

In each of a finite number $T > 0$ of negotiation rounds, the players take turns proposing a partition of the pool between themselves, with player~1 moving first in each round. 
In its turn, a player can accept the latest offer from the other player (\textsc{deal}), end negotiations (\textsc{walk}), or make an offer-revelation combination of the form $(\offer,R)$. 
Offer $\offer \in \{ (\mathbf{p}_1, \mathbf{p}_2) \mid \mathbf{p}_1 + \mathbf{p}_2 = \mathbf{p} \}$ is a proposed partition of the items, with $\mathbf{p}_j$ a vector of $\tau$ non-negative integers representing player~$j$'s share.
Revelation $R\in \{ \T, \F \}$ represents that player's decision to either disclose its signal (\T) in that turn or not (\F). 
We also include a discount factor $\gamma \in (0,1]$ to capture preference for reaching deals sooner. 
Negotiation fails if a player chooses \textsc{walk} in any round $\rho \in\{1,\dotsc,T\}$ or $T$ rounds pass without any player choosing \textsc{deal}. 
In case of failure in round $\rho$, each player~$j$ receives a reward of $\gamma^{\rho} \mathbf{o}_j \cdot \mathbf{v}_j$ from its  outside offer. 
If a proposed partition $(\mathbf{p}_1, \mathbf{p}_2)$ is accepted in round $\rho$, then the reward to~$j$ is $\gamma^{\rho-1} \mathbf{p}_j \cdot \mathbf{v}_j$.
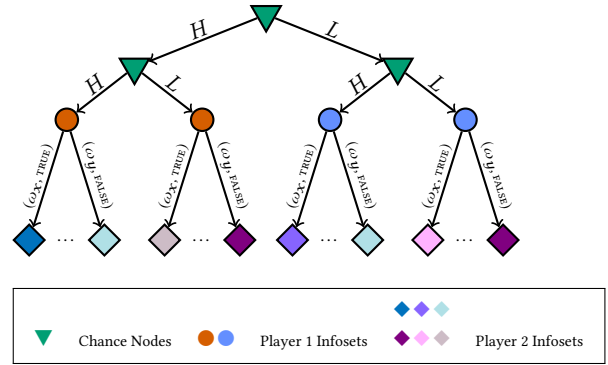
\begin{figure}[h!]
    \centering
\begin{subfigure}[t]{\textwidth}
\begin{tikzpicture}[thick,
    level 1/.style = {level distance = 7mm, sibling distance = 35mm},
    level 2/.style = {level distance = 7mm, sibling distance = 18mm},
    level 3/.style = {level distance = 16mm, sibling distance = 10mm},
    engine/.style = {inner sep = 1pt, above}]
    \node [draw, black, fill={chance}, regular polygon, regular polygon sides=3, rotate=180, inner sep=0.08cm] {}
    [black, ->]

    child {node [draw, black, fill={chance}, regular polygon, regular polygon sides=3, rotate=180, inner sep=0.08cm] (1) {} 
      child {node [draw, black, fill={player1}, circle] (2) {} 
        child {node [draw, black, fill={player2}, diamond] (3) {}
        edge from parent node[engine, sloped] {{\tiny $(\offer_x, \T)$}}}
        child {node [draw, black, fill={player2_2}, diamond] (4) {}
        edge from parent node[engine, sloped] {{\tiny $(\offer_y, \F)$}}}
      edge from parent node[engine, sloped] {$H$}}
      child {node [draw, black, fill={player1}, circle] (6) {} 
        child {node [draw, black, fill={player2_5}, diamond] (7) {}
        edge from parent node[engine, sloped] {{\tiny $(\offer_x, \T)$}}}
        child {node [draw, black, fill={player2_3}, diamond] (8) {}
        edge from parent node[engine, sloped] {{\tiny $(\offer_y, \F)$}}}
      edge from parent node[engine, sloped] {$L$}}
      edge from parent node[engine, sloped] {$H$}}
    child {node [draw, black, fill={chance}, regular polygon, regular polygon sides=3, rotate=180, inner sep=0.08cm] (10) {}
    child {node [draw, black, fill={player1_1}, circle] (11) {} 
        child {node [draw, black, fill={player2_1}, diamond] (12) {}
        edge from parent node[engine, sloped] {{\tiny $(\offer_x, \T)$}}}
        child {node [draw, black, fill={player2_2}, diamond] (13) {}
        edge from parent node[engine, sloped] {{\tiny $(\offer_y, \F)$}}}
      edge from parent node[engine, sloped] {$H$}}
      child {node [draw, black, fill={player1_1}, circle] (15) {} 
        child {node [draw, black, fill={player2_4}, diamond] (16) {}
        edge from parent node[engine, sloped] {{\tiny $(\offer_x, \T)$}}}
        child {node [draw, black, fill={player2_3}, diamond] (17) {}
        edge from parent node[engine, sloped] {{\tiny $(\offer_y, \F)$}}}
      edge from parent node[engine, sloped] {$L$}}
      edge from parent node[engine, sloped] {$L$}
    };

    \path (3) -- node[auto=false]{{\footnotesize \ldots}} (4);
    \path (7) -- node[auto=false]{{\footnotesize \ldots}} (8);
    \path (12) -- node[auto=false]{{\footnotesize \ldots}} (13);
    \path (16) -- node[auto=false]{{\footnotesize \ldots}} (17);
\end{tikzpicture}
\end{subfigure}

\begin{subfigure}[t]{\textwidth}
\begin{flushleft}
\vspace{0.1em}
    \fbox{\begin{tabular}{rlrlrl}
& & & & \scriptsize{\textcolor{player2}{\DiamondSolid} \textcolor{player2_1}{\DiamondSolid} \textcolor{player2_2}{\DiamondSolid}} & \\ 
\scriptsize{\textcolor{chance}{\TriangleDown}} & \scriptsize Chance Nodes &
\scriptsize{\textcolor{player1}{\CircleSolid} \textcolor{player1_1}{\CircleSolid}} & \scriptsize Player 1 Infosets &
\scriptsize{\textcolor{player2_3}{\DiamondSolid} \textcolor{player2_4}{\DiamondSolid} \textcolor{player2_5}{\DiamondSolid}} & \scriptsize Player 2 Infosets
\end{tabular}}
\end{flushleft}
\end{subfigure}

\caption{Part of game tree llustrating the effect of Player 1's decision $R$ on Player 2's infoset structure.
}
\label{fig:DOND_no_offer_reveal}
\end{figure}

Fig.~\ref{fig:DOND_no_offer_reveal} displays a partial extensive-form representation of a simulation of \barg after each player's valuation has already been sampled.
At the first two levels, the simulator samples outside offer signals for the two players from $\{H,L\}$; if $L$ (resp. $H$) is drawn for player~$j$, its actual outside offer is sampled uniformly at random from all possible item vectors $\mathbf{o}$ of cumulative value $\mathbf{o} \cdot \mathbf{v}_j$ above $0$ and at most the threshold $\thresh$ (resp. above $\thresh$ and at most $\bar{V}$). Thus, setting $P_j$ reduces to picking a probability of the signal being $H$ for player~$j$.
Next, player~$1$ chooses an action, comprising an offer $\offer$ and revelation $R$. Player~$2$ now has four distinguishable histories that result when player~$1$ reveals its signal and two non-singleton information sets that result when player~$1$ chooses not to reveal.
The game continues (not shown) with the action of player~$2$, and alternation between the players for another $T-1$ rounds.

\subsection{\gengoof}\label{sec:GenGoof}
\gengoof is parametrized by a positive integer $K$ which determines the number of game rounds ($K-1$) as well as the size of each player's action space ($K$).

First, a discrete stochastic event with $K$ outcomes occurs at the game root, that is, $\abs{X(\emptyset)}=K$.
Let $e_1$ denote the realized outcome of this event.
Player~$1$ observes $e_1$ and chooses one of $K$ actions, say $a^k_1$, from $\cA_1(I(e_1))=\{ a^k_1 \}_{k=1}^K$. Then, player~$2$ observes $e_1$ but not player~$1$'s action, say $a^k_2$, and also chooses one of $K$ actions from $\cA_2(I(e_1 a^k_1)) = \{ a^k_2 \}_{k=1}^K$, ending round~$1$.
Round~$2$ begins with the realization $e_2$ of a second stochastic event with support $X\left(e_1 a^k_1 a^k_2\right) = X(\emptyset) \setminus \{ e_1 \}$. Player~$1$ then observes the history up to and including $e_2$ before choosing one of $K$ actions, followed by player~$2$ who observes all but player~$1$'s second chosen action. 
This process repeats until the final round $K-1$ where a stochastic event with only $2$ possible outcomes occurs, followed by player~1 and player~2 both observing the history until the final stochastic event realization and picking one of $K$ actions each, ending the game. 

To complete the game description, we define the probability distribution of each stochastic event and each leaf utility as follows (these are also game parameters that are hidden from the game analyst).
For each instance of \gengoof, we sample a single $K$-outcome categorical probability distribution uniformly at random for the stochastic event in round~$1$; for $k\in\{2,3,\dots,K-1\}$, the distribution of the round-$k$ stochastic event is obtained by renormalizing the round-$(k-1)$ over the residual support after eliminating the outcome realized in round $(k-1)$. For example, the probability distribution of the round-$2$ stochastic event given that $e_1$ occurred in round $1$ is
\begin{equation*}
    P\left(e_2 \given e_1 a^k_1 a^k_2 \right) = \frac{P(e_2 \given \emptyset)}{\sum\limits_{e' \in X(\emptyset) \setminus \{ e_1 \}} P(e' \given \emptyset)} \quad \forall e_2 \in X(\emptyset) \setminus \{ e_1 \}.
\end{equation*}

For each possible combination of the stochastic event outcome and the two players' action choices in each round of any game instance, we choose a reward for each player uniformly at random from $[0, u_{\max}]$ for a positive real number $u_{\max}$; we set the leaf utility for each player equal to the sum of the player's rewards over all $K - 1$ rounds in the corresponding history. Thus, for every leaf $z \in Z$ and player $j\in\{1,2\}$, $u_j(z) \sim \Udist([0,u_{\max}(K-1)])$ where $\Udist(\cS)$ denotes the uniform distribution over the set $\cS$.
Figure~\ref{fig:gengoof_chap3} in App.~\ref{sec:gengoof4_diag} illustrates the first round of gameplay in a particular instance of \gengoof.


\section{Tree-Exploiting PSRO}\label{sec:method}
Our domain of interest comprises game instances where expanding the full game out as an extensive-form tree for the purpose of game analysis is infeasible. TE-PSRO tackles this challenge by maintaining a coarsened and abstracted (yet extensive-form) version as its empirical game model.
The full game is specified in the form of a gameplay simulator, which is formalized in terms of the world state framework (\S\ref{sec:prelim}).
A key question for TE-PSRO (Fig.~\ref{fig:psro}) is how to translate new best-response results into elements that can be systematically incorporated into an abstracted empirical game model as part of the model augmentation operation.
Our approach bridges the detailed state-space model of the simulator and the simplified game tree using the concept of abstract policies.


\subsection{Abstract Policies}\label{sec:abstraction}

In the underlying game, the space of possible infostates of player~$j$ is given by $\cI_j$, and a policy $\pi_j$ specifies $j$'s action $a \in \cA_j(I)$ for each $I \in \cI_j$.
We represent such policies in our implementation by neural networks,  encoded as a set of weights for a given architecture.
In the empirical game~$\hat{G}$, player~$j$'s information possibilities are described by its infosets~$\hat{\cI_j}$.
In general, there need be no particular structural relationship between $\cI_j$ and~$\hat{\cI_j}$, though typically they will both be defined in terms of a shared set of primitive observations.
We capture the connection by a function $\infomap:\cI_j\to \hat{\cI_j}$, where $\infomap(I_j)$ is the empirical game infoset corresponding to underlying infostate~$I_j$.

Given the distinct formulations, policies $\pi_j$ are executable in the simulator, but cannot be directly interpreted within the framework of empirical game~$\hat{G}$. 
Nevertheless, we can incorporate them as actions in~$\hat{G}$.
For a given policy $\pi^x_j$, we treat the label ``$\pi^x_j$'' as a potentially allowable action for any infoset~$\hat{\cI_j}$. 
From the empirical game perspective, ``$\pi^x_j$'' is an abstract policy.
An overall game-tree strategy specifies an action for every infoset in~$\hat{\cI_j}$.
To execute a game-tree strategy profile in the simulator, we simply trace through the tree, applying selected abstract policies at each infoset.
The selected abstract policy remains in force for player~$j$ until a new infoset is reached where it is $j$'s turn to move.
Though uninterpreted in the game tree itself, these abstract policies have full access to the information state from the simulation needed for execution.

\subsection{Best Response: Deep Reinforcement Learning}\label{sec:best_response}

With the ability to simulate profiles over the empirical game strategy space, we can employ the simulator within a deep RL algorithm to derive best responses~$\piBR_j$.
Our implementation employs the DQN algorithm \citep{mnih2015dqn}, which combines a feed-forward neural network parameterized by $\vartheta$ with temporal difference learning and a second target network to estimate Q-values over time given $I \in \cI_j$.


 As an illustration, we provide details of our deep Q-network for \barg. The input to the neural network representing~$\piBR_j$ for this game is an encoding of player~$j$'s current information state~$I$. 
 $\lceil \log_2(\bar{V}) \rceil$ bits are allotted for player~$j$'s valuation $\mathbf{v}_j[i]$, for each $i \in [\tau]$. One bit is allotted for player~$j$'s outside offer signal. For each player's turn in the game, $\mathbf{p}[i] + 1$ bits are allotted per item type $i \in [\tau]$ to represent a partition of $\mathbf{p}$, plus one bit for the decision to reveal the signal or not. 
 Two bits are allotted to represent the other player's signal: \textsf{00} means no reveal so far, \textsf{01} means~$L$, and \textsf{10} means~$H$. 
 One final bit is allotted to be set to \textsf{1} when negotiations are complete.
 The output of the network is an $\vert \cA_j \vert$-long vector containing the Q-values of each action in $\cA_j$ given the input infostate vector. Our parameter settings, optimized via hyperparameter tuning, are included in App.~\ref{app:dqn}. After successfully training player~$j$'s DQN, the learned weights of $\pi_{\vartheta}$ are saved and mapped to the abstract policy label ``$\pi^x_j$'' in~$\hat{G}$ (\S\ref{sec:abstraction}).

\subsection{Augmenting the Empirical Game Model}\label{sec:best_response_augment}

Given BRs $\piBR$ computed from DRL, the next step is to augment the empirical game~$\hat{G}$ (Fig.~\ref{fig:psro}).
Though an abstract policy is potentially applicable at any point in the game tree, adding $\piBR_j$ to every $I \in \hat{\cI}_j$ could lead to unsustainable growth in~$\hat{G}$.

Our approach is to select a fixed number $M$ of infosets to augment for each player.
Our selection is based on an assessment of the \term{gain}~$\gain$ of playing $\piBR_j$ instead of $\targprof_j$ at candidate infosets $I \in \hat{\cI}_j$ where $\targprof$ is the  BR target at the current TE-PSRO epoch.
Recall (\S\ref{sec:prelim}) that a terminal history in the underlying game is expressed as a sequence of world states and actions: $z = (w^1, a^1, \ldots, w^t)$, and that $R_j(z)$ is the associated reward for player~$j$.
The reach probability of $z$ given that history $h$ of length $t_h$ was reached is
\begin{equation*}
    r(z \mid h, \hat{\bsigma}) = \prod_{\ell = t_h}^{t} \hat{\bsigma}(\infomap(w^{\ell}))(\pi^x) \cdot \mathds{1}\{ \pi^x(w^{\ell}) = a^{\ell} \} \cdot T(w^{\ell}, a^{\ell})(w^{\ell + 1}).
\end{equation*}
The expected payoff to $j$ for playing $\hat{\sigma}_j$ in response to $\targprof_{-j}$ at $I \in \hat{\cI}_j$ is given by
\begin{equation*}
    U_j(\hat{\sigma}_j,\targprof_{-j}, I) = \sum\limits_{\substack{h \in H\mid \infomap(h) \in I}} \sum\limits_{z \in Z} r(z \mid h, \hat{\sigma}_j,\targprof_{-j}) R_j(z).
\end{equation*}

Let $\left.\hat{\bsigma}\right|_{I \rightarrow \piBR_j}$ denote a strategy profile identical to $\hat{\bsigma}$ except that player~$j$ selects $\piBR_j$ at $I$. Gain $K$ is equal to the product of the gain to player~$j$ of $\left.\targprof\right|_{I \rightarrow \piBR_j}$, given that $I$ is reached, and the probability of reaching the set of histories in the underlying game that translate into $I$:
\begin{equation*}
    \gain = r(I, \targprof) \cdot \left( U_j(\left.\targprof\right|_{I \rightarrow \piBR_j}, I) - U_j(\targprof, I) \right).
\end{equation*}

We then perform a softmax selection of $M$ infosets based on the gains $\left[ \gain_I \right]_{I \in \hat{\cI}_j}$.
BR policy $\piBR_j$ is then added as an action edge to each of these infosets.
The process creates new infosets, depending on the observable effects of the abstract policy.
We illustrate how the empirical game tree is extended by this method in App.~\ref{app:te_psro_snapshots}, using \barg as an illustrative example.

Finally, TE-PSRO updates payoff estimates for the augmented~$\hat{G}$ by simulating the strategy combinations that result from the newly added edges and recording the sampled payoffs. 
All epochs are allocated the same total number of gameplay simulations, called the \textit{simulation budget}, distributed equally among all \textit{new} strategy profiles. Thus, the number of samples per profile is fixed based on the TE-PSRO epoch, independent of the choice of~$M$.


\section{Computing Refined Nash Equilibria}\label{sec:SPE}

Tree-based game models afford consideration of solution concepts specific to the extensive form.
We specifically investigate the use of \term{subgame perfect equilibrium} (SPE), a refinement of NE that rules out solutions containing non-credible threats.
To make use of SPE, we need a definition that applies to games of imperfect information, and an algorithm that computes such solutions.

\begin{definition}\label{def:subgame}
A \term{subgame} of game $G$ is a directed rooted subtree given by $G'= \langle N, H', V', \{\mathcal{I}'_j\}_{j=0}^n, \{A'_j\}_{j=1}^n, X', P', u' \rangle$ satisfying the following:
\begin{itemize}
    \item The root $h'$ of tree $H'$ must be the only node in its information set.
    \item As a subtree of $H$, $H'$ must include all nodes in $H$ that succeed~$h'$.
    \item For any $j \in N$ and for all $I \in \mathcal{I}_j$, if $I \in \mathcal{I}'_j$, then the nodes $h \in I$ must all be part of $H'$; if $I \notin \mathcal{I}'_j$, then all its nodes must be part of $H \setminus H'$.
    \item $V'$, $\{\mathcal{I}'_j\}_{j=0}^n$, $\{A'_j\}_{j=1}^n$, $X'$, $P'$, and $u'$ are restrictions to $H'$ of $V$, $\{\mathcal{I}_j\}_{j=0}^n$, $\{A_j\}_{j=1}^n$, $X$, $P$, and $u$, respectively.
\end{itemize}
\end{definition}


\begin{definition}[\citep{selten75}]\label{def:SPE}
A \term{subgame perfect equilibrium} (SPE) of game $G$ is an NE of $G$ that also induces NE play in each of $G$'s subgames. 
\end{definition}
In finite perfect-information EFGs, an SPE always exists in pure strategies and can be readily computed using the classic backward induction approach. 
With imperfect information, however, a subgame may not admit a pure-strategy NE at all. 
\citet{kaminski19} proposed the \term{generalized backward induction} (GBI) approach for finding the set of SPE for a potentially infinite game of imperfect information.
A key feature of GBI is the re-expression of the game tree as a set of its proper subgames organized by their roots. 
Other crucial implementation details are not fully specified in the original article, in particular how to identify NE of subgames that include non-singleton information sets; a na\"{\i}ve implementation using exhaustive enumeration of strategy profiles for combinations of subgames has a runtime that is exponential in game size.

We provide a practical, modular algorithm for finding an SPE via GBI in a finite, imperfect-information EFG.
Our algorithm combines dynamic programming with a Nash solver subroutine, using \citeauthor{kaminski19}'s [\citeyear{kaminski19}] idea of organizing the game into subgames. 
Alg.~\ref{alg:compute_spe} presents our method.
\textsc{ComputeSPE} uses subroutines described here at a high level (see App.~\ref{app:SPE} for full pseudocode). 

\begin{algorithm}[!htb]
\small
\caption{: \textsc{ComputeSPE}}
\label{alg:compute_spe}
\begin{algorithmic}[1]
\Require{Input game $G$}

\State{$\Psi \gets \textsc{GetSubgameRoots}(G)$}
\State{$\ell \gets$ height of $h_0$ in $\Psi$}
\State{$\{ \Theta_k \}_{k = 1}^{\ell} \gets \textsc{GetSubgameGroups}(G, \Psi, \ell)$}

\State{$\bm{\sigma}^{SPE} \gets \textsc{GetInitialSPE}(G, \Theta_1)$}
\For{$1 < k \leq \ell$}
\For{$\theta \in \Theta_k$}
\State{Extract $\left.\bm{\sigma}^{SPE}\right|_{G_{\theta}} \gets \{ \bm{\sigma}^{SPE}(I) \mid I \in G_{\theta} \cap \bm{\sigma}^{SPE} \}$}
\State{$\bm{\sigma}^{SPE} \gets \bm{\sigma}^{SPE} \cup \textsc{NashSolver}(G_{\theta}, \left.\bm{\sigma}^{SPE}\right|_{G_{\theta}})$}
\EndFor
\EndFor

\Return $\bm{\sigma}^{SPE}$
\end{algorithmic}
\end{algorithm}
 
 We first call \textsc{GetSubgameRoots} to find the roots of all subgames in the input game $G$ and arrange them into a tree $\Psi$ rooted at $h_0$, the root of $G$. 
 A root in $\Psi$ has a \term{height} $1 \leq k \leq \ell$ where the subgames closest to the leaves of $G$ have height~1 and $h_0$ has height~$\ell$.
 To find the roots, it is sufficient to check which nodes in $H$ are roots of subtrees satisfying the conditions of Definition~\ref{def:subgame}.
 \textsc{GetSubgameGroups} collects all subgame roots in $\Psi$ at height~$k$ into a set~$\Theta_k$, for $1 \leq k \leq \ell$. 
 Then, we use dynamic programming to iterate over the subgames of each $\Theta_k$ and solve each subgame at height $k$ directly via a chosen \textsc{NashSolver}. The union $\bm{\sigma}^{SPE}$ of all SPE found for the subgames in $G$ with height less than $k$ is updated with each new partial SPE. $G_{\theta}$ denotes the subgame rooted at node $\theta \in \Theta_k$. 
 The union of all SPE across all subgames in $G$ by definition must be the SPE of $G$. In order to avoid overwriting the SPE that have been computed for any subgames at smaller heights in $G_{\theta}$, we pass the partial SPE $\left.\bm{\sigma}^{SPE}\right|_{G_{\theta}}$ in as input to \textsc{NashSolver} and restrict the solver to find a solution only for the information sets within~$G_{\theta}$ that are not already included in $\left.\bm{\sigma}^{SPE}\right|_{G_{\theta}}$. 
 This ensures that the runtime of \textsc{ComputeSPE} is linear with respect to the size of $G$ and thus scalable modulo the runtime of \textsc{NashSolver} (see App.~\ref{app:proofs} for runtime analysis). 
 For our experiments (\S\ref{sec:experiments}), we devised an adaptation of the counterfactual regret (CFR) minimization algorithm~\citep{cfrm07}, called $\textsc{SubgameCFR}$, as our \textsc{NashSolver}.


\section{Experiments}\label{sec:experiments}
We now report experiments that we conducted to evaluate our TE-PSRO approach by applying it to the two games described in \S\ref{sec:game_description}.

\subsection{Parameter settings}
For \barg (\S\ref{sec:DOND}), we set $\tau=3$, $\bar{V}=10$, $\thresh=5$, $\gamma=0.99$, $T=5$, and $n\in\{5,6,7\}$. 
We generated five unique sets of the remaining parameters $\mathbf{p}$, $(\mathbf{v}_1, \mathbf{v}_2)$, $P_1$, $P_2$ uniformly at random from their respective supports in order to evaluate TE-PSRO's performance on a variety of game instances. For \gengoof (\S\ref{sec:GenGoof}), we set $K=4$, $u_{\max}=10$, calling this instance \gengoofK{4}.  The simulator budget was $100$ samples for \barg and $200$ samples for \gengoofK{4}.

We ran all experiments on our local computing cluster using a single core. 
Runtime and memory requirements depend on the choice of $M \in \{ 1, 2, 4, 8, 16 \}$, which determines the rate of growth of~$\hat{G}$ across TE-PSRO epochs (see App.~\ref{app:great_lakes} for details). 
Unless otherwise stated, every experiment was performed for five randomly seeded trials for each setting. Error bars in our plots correspond to a 95\% confidence interval.

\subsection{Results}\label{sec:exp_results}
Our first set of experiments assesses space requirements for the empirical game $\hat{G}$ as a function of the number $M$ of infosets augmented per epoch. 
At each epoch of TE-PSRO, we recorded the total number of information sets across players in $\hat{G}$ and the memory required by the emprical model.
We report the average empirical game size for each of the two games studied in terms of the number of player information sets of all players and memory required in megabytes (MB) 
in Fig.~\ref{fig:game_size} and Fig.~\ref{fig:game_size_app} in App.~\ref{app:game_size} respectively, for representative values of $M$; each curve for \barg is averaged over 100 trials per value of $M$ across all five sets of bargaining parameters. The broad takeaway from all plots in this set is the following. Although the rate of increase in the size of the $\hat{G}$ steepens with $M$ in the plots, its size is still manageable after many epochs of TE-PSRO. If we had added a new policy to all information sets rather than to only a subset of size $M$, $\hat{G}$ would grow to as many as $3000$ or $4000$ information sets after only five epochs of TE-PSRO, which is an unsustainable trajectory. See App.~\ref{app:game_size} for further insights on the difference between the two games.
\begin{figure}[h]
    \centering
    \begin{subfigure}[b]{0.36\textwidth}
     \caption{\barg}
    \includegraphics[width=\textwidth]{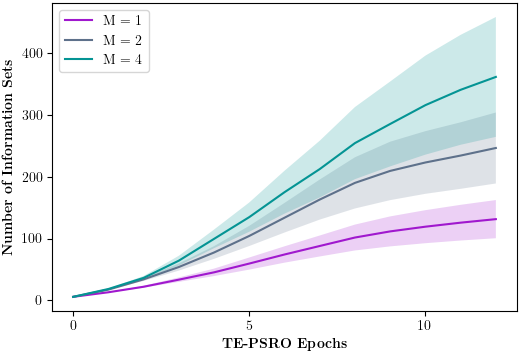}
    \label{fig:infosets_barg}
    \end{subfigure}
    \begin{subfigure}[b]{0.36\textwidth}
    \caption{\gengoofK{4}}
    \includegraphics[width=\textwidth]{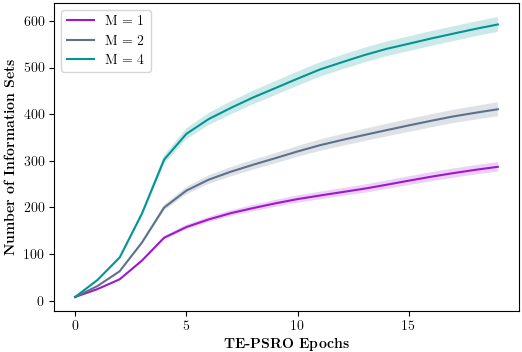}
     \label{fig:infosets_gengoof}
    \end{subfigure}
    \caption{Total number of information sets of both players in empirical game $\hat{G}$ over the course of TE-PSRO's runtime, averaged over all combinations of all parameters (except $M$) and seeds.}
    \label{fig:game_size}
\end{figure}

Our second set of experiments provides evidence for the effectiveness of our algorithm \textsc{ComputeSPE}(\S\ref{sec:SPE}) as well as the non-triviality of obtaining an SPE of an imperfect-information extensive-form game.
We ran two suites of TE-PSRO on each of \barg and \gengoofK{4}: each suite consisted of 50 trials for each $M$ setting, using NE as the MSS for 25 trials and SPE as the MSS for the remaining~25.
In one suite, we evaluated intermediate and final $\hat{G}$ (EVAL in Fig.~\ref{fig:psro}) by NE computed using CFR, and the other by SPE using Alg.~\ref{alg:compute_spe}.
We computed the regret of the respective solutions with respect to each subgame of $\hat{G}$ and reported the maximum over subgames, that is, the \term{worst-case subgame regret}; a solution with a lower value of this quantity is a better SPE approximation. Thus,  
Figs.~\ref{fig:subgame_regret_SPE} and~\ref{fig:subgame_regret_abs4_SPE} verify that Alg.~\ref{alg:compute_spe} does indeed produce an approximate SPE
Figs.~\ref{fig:subgame_regret_NE} and~\ref{fig:subgame_regret_abs4_NE} demonstrate that our regular Nash solver does not happen to stumble upon NE that are also subgame-perfect. 
Additionally, the worst-case subgame regret increases with the complexity of $\hat{G}$, reflected in both setting of $M$ and epochs of TE-PSRO. Note that these experiments are not for gauging the quality of the models produced by TE-PSRO; instead, TE-PSRO is used to generate a sequence of empirical games of increasing size and complexity (in terms of the number of non-singleton information sets) that serve as more and more challenging test cases for our game-solving algorithms.


\begin{figure*}[ht!]
    \centering
    \begin{subfigure}[b]{0.36\textwidth}
    \caption{\barg: Empirical NE}
    \includegraphics[width=\textwidth]{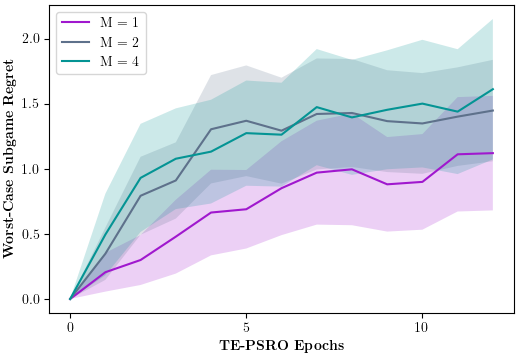}
    \label{fig:subgame_regret_NE}
    \end{subfigure}\hspace{50pt}
    \begin{subfigure}[b]{0.36\textwidth}
     \caption{\barg: Empirical SPE}
    \includegraphics[width=\textwidth]{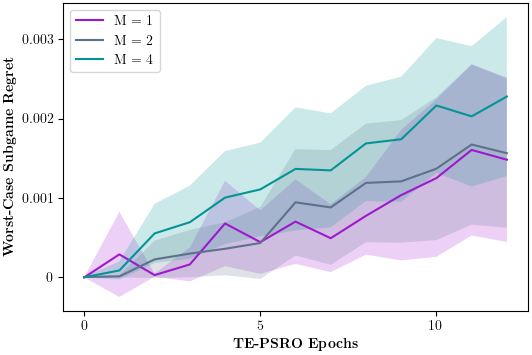}
     \label{fig:subgame_regret_SPE}
    \end{subfigure}\\
    \begin{subfigure}[b]{0.36\textwidth}
     \caption{\gengoof{4}: Empirical NE}
    \includegraphics[width=\textwidth]{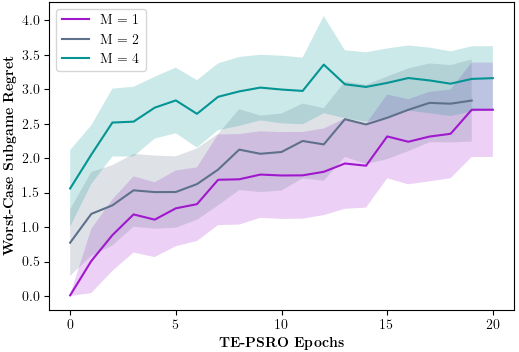}
    \label{fig:subgame_regret_abs4_NE}
    \end{subfigure}\hspace{50pt}
    \begin{subfigure}[b]{0.36\textwidth}
    \caption{\gengoof{4}: Empirical SPE}
    \includegraphics[width=\textwidth]{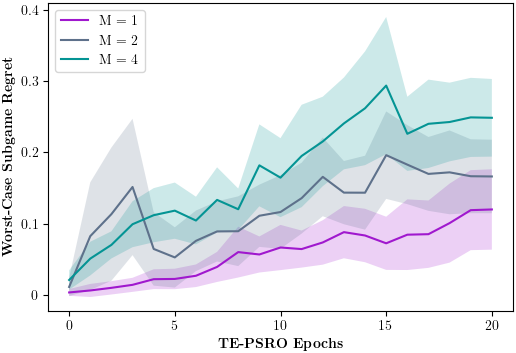}
     \label{fig:subgame_regret_abs4_SPE}
    \end{subfigure}
    \label{fig:subgame_regret_abs4_chap6}
    \caption{Average worst-case subgame regret of NE  and SPE  solutions to empirical game $\hat{G}$ for the two games studied.
    Note that the scale of vertical axis of (b) is finer than that of (a) by a factor $\approx10^3$ while the corresponding factor $\approx10$ for (d) and (c). 
    }
    \label{fig:subgame_regret}
\end{figure*}

Our final experiment set characterizes TE-PSRO performance in terms of the MSS choice (SPE vs. NE) and different values of~$M$. 
To compare  MSS choices, we computed the profile regret (\S\ref{sec:prelim}) with respect to the underlying game of the solution $\bsigma^*$ returned by EVAL in each epoch of TE-PSRO epoch; we used both NE and SPE as EVAL, giving us two sets of comparison metrics. 

Fig.~\ref{fig:true_regret} depicts our average regret results for \barg under various settings.
We ran TE-PSRO for 25 trials per value of $M$ and four combinations of choices for EVAL and MSS. 
In each trial, TE-PSRO was allowed to run for at most 30 epochs, terminating early when the computed best responses did not yield an improvement greater than 0.1 over the current solution $\bsigma^*$.
Fig.~\ref{fig:regret_M1} shows that TE-PSRO outperforms the normal-form version NF-PSRO, regardless of EVAL/MSS choice, even for $M = 1$.
Figs.~\ref{fig:regret_M8_no_NF} and~\ref{fig:regret_M4_no_NF} show that SPE beats NE as an MSS for $M \in \{ 4, 8 \}$, converging faster to near-zero regret regardless of EVAL.
Finally, Fig.~\ref{fig:regret_spe_mss_ne_eval} compares results for different $M$ settings, with NE as EVAL and SPE as MSS. 
The main takeaway is that intermediate values $M=4$ and $M=8$ outperform lower and higher settings.
Intuitively, a higher $M$ produces $\hat{G}$ with broader coverage earlier, but with a fixed sampling budget, each tree path is estimated less accurately, resulting in a non-monotonic performance with respect to $M$. 
App.~\ref{app:expts} presents the full set of plots over combinations of EVAL, MSS, and $M$.
\begin{figure*}[htb!]
    \centering
    \begin{subfigure}[b]{0.36\textwidth}
     \caption{$M = 1$, varying EVAL and MSS, NF-PSRO as baseline}
    \includegraphics[width=\textwidth]{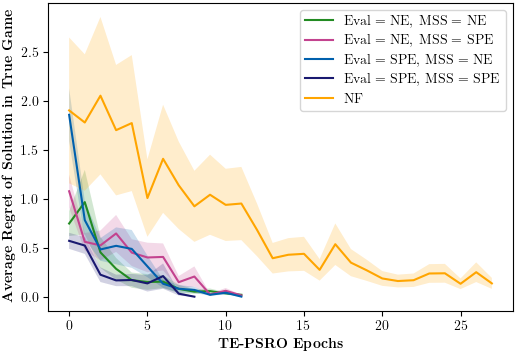}
    \label{fig:regret_M1}
    \end{subfigure}\hspace{50pt}
    \begin{subfigure}[b]{0.36\textwidth}
    \caption{$M = 8$, varying EVAL and MSS}
    \includegraphics[width=\textwidth]{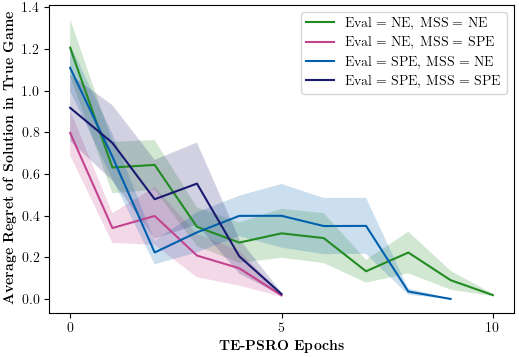}
    \label{fig:regret_M8_no_NF}
    \end{subfigure}
    \begin{subfigure}[b]{0.36\textwidth}
    \caption{$M = 4$, varying EVAL and MSS}
    \includegraphics[width=\textwidth]{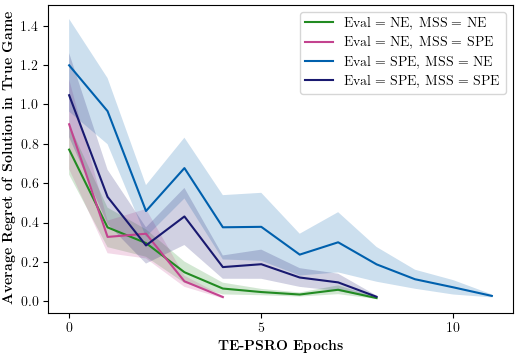}
    \label{fig:regret_M4_no_NF}
    \end{subfigure}\hspace{50pt}
    \begin{subfigure}[b]{0.36\textwidth}
    \caption{NE for EVAL, SPE for MSS, varying $M$}
    \includegraphics[width=\textwidth]{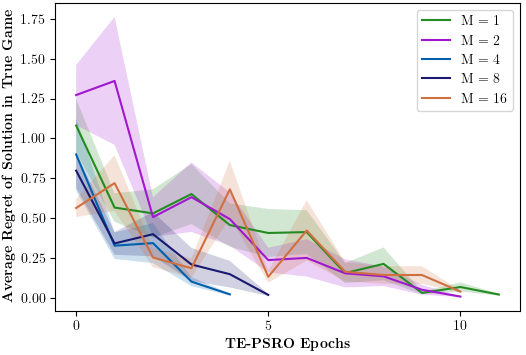}
    \label{fig:regret_spe_mss_ne_eval}
    \end{subfigure}
    \caption{Average regret of solution $\bsigma^*$ of empirical game for \barg over iterations of TE-PSRO, using NE or SPE as the MSS or EVAL and different values of $M$.}
    \label{fig:true_regret}
\end{figure*}

We assessed the statistical significance of the MSS comparisons for \barg using a permutation test, for each setting of $M$ and EVAL. 
Our figure of merit is the area $\cA_\text{MSS}$ under the regret curve, 
calculated starting from TE-PSRO epoch $5$ to avoid the noisy startup phase.
Our test statistic is $\Delta_\text{MSS} = \cA_{\textsc{NE}} - \cA_{\textsc{SPE}}$, resampled over 1000 permutations of the MSSs. 
The $p$-value is the fraction of times the difference under permutation (i.e., a null hypothesis that the MSSs are equally effective) is greater than the observed $\Delta_\text{MSS}$.
For $M = 4$, the superiority of SPE as MSS is significant ($p = 0.007$) with NE as EVAL. 
For $M = 8$, SPE as MSS is significantly better both for NE ($p = 0.006$) and SPE ($p = 0.021$) as EVAL. 
The results are less consistent and less significant for non-optimal $M$ values (App.~\ref{app:p-values}).

For experiments on \gengoofK{4}, we additionally constrained the space of empirical game models induced by TE-PSRO by coarsening away the stochastic events in the last one  two rounds of the full three-round underlying game. We indicate this by $IR$ that stands for \term{included rounds}; it may take values $[0]$, $[0,1]$ or $[0,1,2]$ indicating that each empirical game tree can include (a) stochastic event(s) only in its first round, first two rounds, and all three rounds respectively. 
In Fig.~\ref{fig:abstract_spe_mss_chap8}, we see that $IR = [0]$ and $IR = [0, 1]$ tended to yield the best performance regardless of MSS and that, for both of these settings, SPE outperformed NE as the MSS. App.~\ref{app:IR_gengoof} offers insights on how $IR$ and $M$ jointly impact TE-PSRO performance.

\begin{figure*}[htp!]
    \centering
    \begin{subfigure}[b]{0.36\textwidth}
    \caption{$M = 2$ \label{fig:spe_abs4_M2}}
    \includegraphics[width=\textwidth]{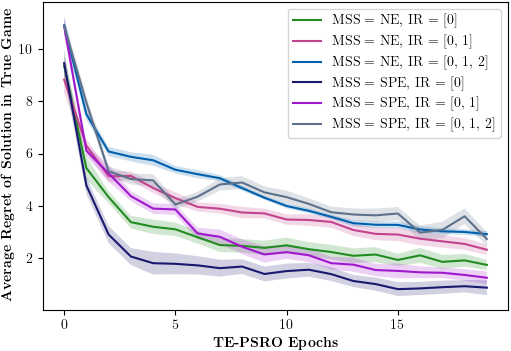}
    \end{subfigure}\hspace{50pt}
    \begin{subfigure}[b]{0.36\textwidth}
     \caption{$M = 4$ \label{fig:spe_abs4_M4}}
    \includegraphics[width=\textwidth]{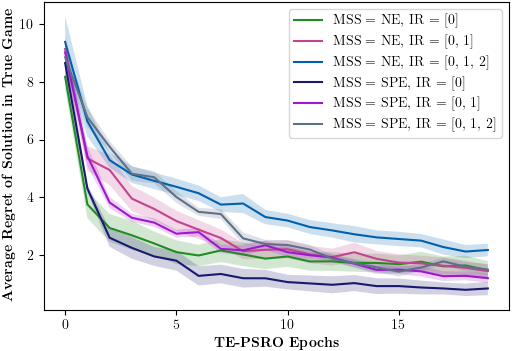}
    \end{subfigure}
    \caption{Average regret of $\bsigma^*$ evaluated in \gengoofK4\ over the course of TE-PSRO's runtime, using NE or SPE as the MSS.}
    \label{fig:abstract_spe_mss_chap8}
\end{figure*}


\section{Conclusions}\label{sec:conclusion}

We introduced multiple extensions of Tree-Exploiting PSRO, enabling its application to complex games of imperfect information. 
Our main innovation is the treatment of best responses computed by DRL as abstract policies, incorporated as actions in the empirical game tree.
To manage growth of the empirical game as BRs are generated over the course of TE-PSRO, we introduced a hyperparameter~$M$, which controls the number of infosets that can be expanded per epoch.
Finally, we demonstrated that having an extensive-form empirical game model can be leveraged in the form of new meta-strategy solvers based on Nash refinements. 
Toward that end, we developed a modular algorithm for identifying SPE solutions in imperfect-information games. We demonstrated these methods on two carefully constructed complex games, featuring multiple rounds of offer/counteroffer with signaling options.
We showed that TE-PSRO easily outperforms normal-form PSRO in this environment, and that intermediate values of $M$ perform best. 

Particularly intriguing is our finding that exploiting Nash refinement in an MSS offers promise for improving strategy exploration.
Even when the goal is not to find a subgame-perfect solution (i.e., EVAL is NE rather than SPE), targeting best response to SPE rather than NE can be beneficial.
Further work will be required to confirm the scope of and understand the reasons for this advantage.
One intuitive explanation is that empirical game equilibria containing non-credible threats may be particularly exploitable in the underlying game, and thus not the most promising lines to pursue in strategy exploration.


\begin{acks}
This work was supported in part by a grant from the Effective Altruism Foundation and by the US National Science Foundation under CRII Award 2153184.
\end{acks}



\clearpage
\bibliographystyle{ACM-Reference-Format} 
\bibliography{aamas2025}

\newpage
\appendix
\onecolumn

\section*{Appendices for AAMAS25 Submission 917}

\section{Illustration of Gameplay in \gengoofK{4}}\label{sec:gengoof4_diag}
\begin{figure}[htbp]
\centering
\begin{subfigure}[t]{0.95\textwidth}
\centering
\begin{tikzpicture}[thick,
    level 1/.style = {level distance = 15mm, sibling distance = 22mm},
    level 2/.style = {level distance = 15mm, sibling distance = 26mm},
    level 3/.style = {level distance = 15mm, sibling distance = 32mm},
    level 4/.style = {level distance = 18mm, sibling distance = 12mm},
    engine/.style = {inner sep = 1pt, above}]
    \node [draw, black, fill={gg_chance}, regular polygon, regular polygon sides=3, rotate=180, inner sep=0.10cm] {}
    [black, ->]
    child { node [draw, black, fill={gg_player1}, circle] (1) {} 
        child {node [draw, black, fill={gg_player2}, diamond] (2) {} 
            child {node [draw, black, fill={gg_chance}, regular polygon, regular polygon sides=3, rotate=180, inner sep=0.10cm] (3) {}
                child {node [draw, black, fill={gg_player1_4}, circle] (4) {}
                edge from parent node[engine, left] {$B$}}
                child {node [draw, black, fill={gg_player1_5}, circle] (5) {}
                edge from parent node[engine, left] {$C$}}
                child {node [draw, black, fill={gg_player1_6}, circle] (6) {}
                edge from parent node[engine, right] {$D$}}
            edge from parent node[engine, sloped] {$a^1_2$}}
            child {node [draw, black, fill={gg_chance}, regular polygon, regular polygon sides=3, rotate=180, inner sep=0.10cm] (7) {}
                child {node [draw, black, fill={gg_player1_7}, circle] (8) {}
                edge from parent node[engine, left] {$B$}}
                child {node [draw, black, fill={gg_player1_8}, circle] (9) {}
                edge from parent node[engine, left] {$C$}}
                child {node [draw, black, fill={gg_player1_9}, circle] (10) {}
                edge from parent node[engine, right] {$D$}}
        edge from parent node[engine, sloped] {$a^4_2$}}
        edge from parent node[engine, sloped] {$a^1_1$}}
        child {node [draw, black, fill={gg_player2}, diamond] (11) {} 
        edge from parent node[engine, sloped] {$a^4_1$}}
    edge from parent node[engine, sloped] {$A$}}
    child { node [draw, black, fill={gg_player1_1}, circle] (12) {} 
    edge from parent node[engine, left] {$B$}}
    child { node [draw, black, fill={gg_player1_2}, circle] (13) {} 
    edge from parent node[engine, right] {$C$}}
    child { node [draw, black, fill={gg_player1_3}, circle] (14) {} 
        child {node [draw, black, fill={gg_player2_1}, diamond] (17) {} 
        edge from parent node[engine, sloped] {$a^1_1$}}
        child {node [draw, black, fill={gg_player2_1}, diamond] (26) {} 
            child {node [draw, black, fill={gg_chance}, regular polygon, regular polygon sides=3, rotate=180, inner sep=0.10cm] (18) {}
        child {node [draw, black, fill={gg_player1_10}, circle] (19) {}
        edge from parent node[engine, left] {$A$}}
        child {node [draw, black, fill={gg_player1_11}, circle] (20) {}
        edge from parent node[engine, left] {$B$}}
        child {node [draw, black, fill={gg_player1_12}, circle] (21) {}
        edge from parent node[engine, right] {$C$}}
    edge from parent node[engine, sloped] {$a^1_2$}}
    child {node [draw, black, fill={gg_chance}, regular polygon, regular polygon sides=3, rotate=180, inner sep=0.10cm] (22) {}
        child {node [draw, black, fill={gg_player1_13}, circle] (23) {}
        edge from parent node[engine, left] {$A$}}
        child {node [draw, black, fill={gg_player1_14}, circle] (24) {}
        edge from parent node[engine, left] {$B$}}
        child {node [draw, black, fill={gg_player1_15}, circle] (25) {}
        edge from parent node[engine, right] {$C$}}
edge from parent node[engine, sloped] {$a^4_2$}}
        edge from parent node[engine, sloped] {$a^4_1$}}
    edge from parent node[engine, sloped] {$D$}
    };
    \node[below=1cm of 12] (15) {};
    \node[below=1cm of 13] (16) {};
    \node[below=1cm of 11] (27) {};
    \node[below=1cm of 17] (28) {};
    \node[below=1cm of 4] (29) {};
    \node[below=1cm of 5] (30) {};
    \node[below=1cm of 6] (31) {};
    \node[below=1cm of 8] (32) {};
    \node[below=1cm of 9] (33) {};
    \node[below=1cm of 10] (34) {};
    \node[below=1cm of 19] (35) {};
    \node[below=1cm of 20] (36) {};
    \node[below=1cm of 21] (37) {};
    \node[below=1cm of 23] (38) {};
    \node[below=1cm of 24] (39) {};
    \node[below=1cm of 25] (40) {};

    \path (12) -- node[auto=false]{{\LARGE \vdots}} (15);
    \path (13) -- node[auto=false]{{\LARGE \vdots}} (16);
    \path (2) -- node[auto=false]{{\LARGE \ldots}} (11);
    \path (3) -- node[auto=false]{{\LARGE \ldots}} (7);
    \path (11) -- node[auto=false]{{\LARGE \vdots}} (27);
    \path (18) -- node[auto=false]{{\LARGE \ldots}} (22);
    \path (17) -- node[auto=false]{{\LARGE \vdots}} (28);
    \path (4) -- node[auto=false]{{\LARGE \vdots}} (29);
    \path (5) -- node[auto=false]{{\LARGE \vdots}} (30);
    \path (6) -- node[auto=false]{{\LARGE \vdots}} (31);
    \path (8) -- node[auto=false]{{\LARGE \vdots}} (32);
    \path (9) -- node[auto=false]{{\LARGE \vdots}} (33);
    \path (10) -- node[auto=false]{{\LARGE \vdots}} (34);
    \path (19) -- node[auto=false]{{\LARGE \vdots}} (35);
    \path (20) -- node[auto=false]{{\LARGE \vdots}} (36);
    \path (21) -- node[auto=false]{{\LARGE \vdots}} (37);
    \path (23) -- node[auto=false]{{\LARGE \vdots}} (38);
    \path (24) -- node[auto=false]{{\LARGE \vdots}} (39);
    \path (25) -- node[auto=false]{{\LARGE \vdots}} (40);
    
\end{tikzpicture}
\end{subfigure}
\begin{subfigure}[t]{0.9\textwidth}
\centering
\fbox{\begin{tabular}{rr}
\small{\textcolor{gg_chance}{\TriangleDown}} & \footnotesize Chance Nodes \\
\small{\textcolor{gg_player1}{\CircleSolid} \textcolor{gg_player1_1}{\CircleSolid} \textcolor{gg_player1_2}{\CircleSolid} \textcolor{gg_player1_3}{\CircleSolid} \textcolor{gg_player1_4}{\CircleSolid} \textcolor{gg_player1_5}{\CircleSolid} \textcolor{gg_player1_6}{\CircleSolid} \textcolor{gg_player1_7}{\CircleSolid}} \\ 
\small{\textcolor{gg_player1_8}{\CircleSolid} \textcolor{gg_player1_9}{\CircleSolid} \textcolor{gg_player1_10}{\CircleSolid} \textcolor{gg_player1_11}{\CircleSolid} \textcolor{gg_player1_12}{\CircleSolid} \textcolor{gg_player1_13}{\CircleSolid} \textcolor{gg_player1_14}{\CircleSolid} \textcolor{gg_player1_15}{\CircleSolid}} & \footnotesize Player 1 Infosets \\
\small{\textcolor{gg_player2}{\DiamondSolid} \textcolor{gg_player2_1}{\DiamondSolid}} & \footnotesize Player 2 Infosets
\end{tabular}}
    \end{subfigure}
    \caption{Extensive-form representation of the first round of \gengoof{4}, an instance of \gengoof  where the number of stochastic outcomes at the start of the game is $K=4$, and the total number of game rounds is $3$.}
    \label{fig:gengoof_chap3}
\end{figure}
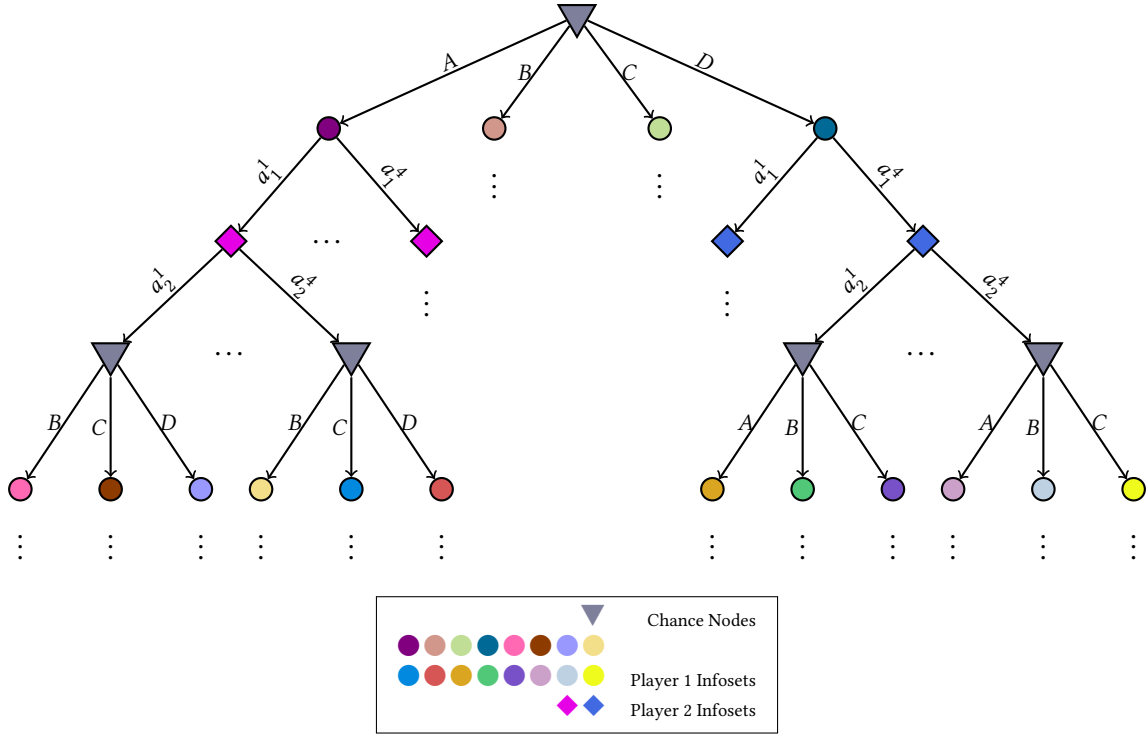

\section{Illustrations of Controlled Empirical Game Tree Expansion With New Best Responses}\label{app:te_psro_snapshots}
Recall the approach introduced in Section~\ref{sec:abstraction}, where each best response policy with learned DQN weights $\vartheta$ is represented with an abstract symbol, and new actions are added to only $M$ of each player's information sets in $\hat{G}$. We use the symbols $\pi_x$ for $x \in \{0, 1, 2, \ldots \}$ in a slight abuse of notation to distinguish between policy labels for different learned weights instead $\pi_{\vartheta_x}$. We now articulate an example of how the empirical game tree might now grow through the new paths generated for $M = 1$ using the sequential bargaining game introduced earlier as a running example. The initial strategy profile of the empirical game tree is set for the information sets that comprise both players' first turns in the negotiations. An example is depicted in Figure~\ref{fig:empirical_game_iter0} for an initial strategy profile where $\bsigma_1$ at both of player 1's information sets $(H,)$ and $(L,)$ returns the sole action $(\pi_0, \F)$ and $\bsigma_2$ at player 2's information sets $(H, (\pi_0, \F))$ and $(L, (\pi_0, \F))$ returns the sole action $(\pi_0, \T)$. The choice to reveal the outside offer signal is selected randomly for both players. The policy $\pi_0$ for player $j$ corresponds to an initial random weight setting $\vartheta_0^j$ of that player's DQN (i.e. $\pi_0$ for player 1 is not identical to $\pi_0$ for player 2). We abuse notation slightly in Figure~\ref{fig:empirical_game_iter0} by including the revelation $R$ that is part of the action outputted by $\vartheta_0^j$ for each player $j$ in the edges of $\hat{G}$ in addition to $\pi_0$ so that the effect of player 1's choice of $R$ on player 2's information sets is captured visually.

The best response policy for each player $\pi_1$ with weights $\vartheta_1$ is then added to the action space of $M$ information sets in $\hat{\cI_j}$ per player, where $M$ is defined as in Section~\ref{sec:abstraction} and is set to 1 for this example. As an additional check on $\hat{G}$'s growth, the number of negotiation turns explicitly included so far in the empirical game tree is increased by at most one during the expansion of the tree, depending on where the best responses are added. It is important to note that we require the revelation $R$ of the action outputted by the weights $\vartheta_x$ of policy $\pi_x$ to be identical for all infostates represented by a given $I$ in $\hat{G}$.

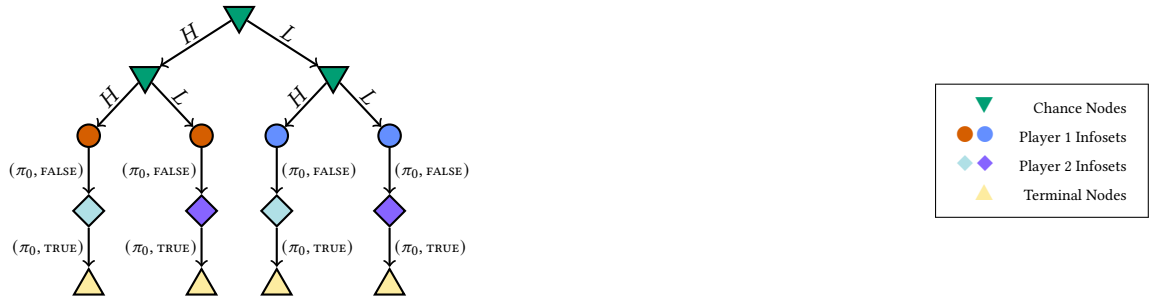
\begin{figure}[ht!]
\centering
\begin{subfigure}[c]{0.45\textwidth}
\begin{flushleft}
\begin{tikzpicture}[thick,
    level 1/.style = {level distance = 8mm, sibling distance = 25mm},
    level 2/.style = {level distance = 8mm, sibling distance = 15mm},
    level 3/.style = {level distance = 10mm, sibling distance = 15mm},
    level 4/.style = {level distance = 10mm, sibling distance = 8mm},
    engine/.style = {inner sep = 1pt, above}]
    \node [draw, black, fill={chance}, regular polygon, regular polygon sides=3, rotate=180, inner sep=0.08cm] {}
    [black, ->]

    child {node [draw, black, fill={chance}, regular polygon, regular polygon sides=3, rotate=180, inner sep=0.08cm] (1) {} 
      child {node [draw, black, fill={player1}, circle] (2) {} 
        child {node [draw, black, fill={player2_2}, diamond] (3) {}
          child {node [draw, black, fill={terminal}, regular polygon, regular polygon sides=3, inner sep=0.08cm] (4) {}
          edge from parent node[engine, left] {{\scriptsize $(\pi_0, \T)$}}}
        edge from parent node[engine, left] {{\scriptsize $(\pi_0, \F)$}}}
      edge from parent node[engine, sloped] {$H$}}
      child {node [draw, black, fill={player1}, circle] (5) {} 
        child {node [draw, black, fill={player2_1}, diamond] (3) {}
          child {node [draw, black, fill={terminal}, regular polygon, regular polygon sides=3, inner sep=0.08cm] (4) {}
          edge from parent node[engine, left] {{\scriptsize $(\pi_0, \T)$}}}
        edge from parent node[engine, left] {{\scriptsize $(\pi_0, \F)$}}}
      edge from parent node[engine, sloped] {$L$}}
    edge from parent node[engine, sloped] {$H$}}
    child {node [draw, black, fill={chance}, regular polygon, regular polygon sides=3, rotate=180, inner sep=0.08cm] (1) {} 
      child {node [draw, black, fill={player1_1}, circle] (2) {} 
        child {node [draw, black, fill={player2_2}, diamond] (3) {}
          child {node [draw, black, fill={terminal}, regular polygon, regular polygon sides=3, inner sep=0.08cm] (4) {}
          edge from parent node[engine, right] {{\scriptsize $(\pi_0, \T)$}}}
        edge from parent node[engine, right] {{\scriptsize $(\pi_0, \F)$}}}
      edge from parent node[engine, sloped] {$H$}}
      child {node [draw, black, fill={player1_1}, circle] (5) {} 
        child {node [draw, black, fill={player2_1}, diamond] (3) {}
          child {node [draw, black, fill={terminal}, regular polygon, regular polygon sides=3, inner sep=0.08cm] (4) {}
          edge from parent node[engine, right] {{\scriptsize $(\pi_0, \T)$}}}
        edge from parent node[engine, right] {{\scriptsize $(\pi_0, \F)$}}}
      edge from parent node[engine, sloped] {$L$}}
    edge from parent node[engine, sloped] {$L$}
    };

\end{tikzpicture}
\end{flushleft}
\end{subfigure}
\begin{subfigure}[c]{0.4\textwidth}
\begin{flushright}
    \fbox{\begin{tabular}{rr}
\scriptsize{\textcolor{chance}{\TriangleDown}} & \scriptsize Chance Nodes \\
\scriptsize{\textcolor{player1}{\CircleSolid} \textcolor{player1_1}{\CircleSolid}} & \scriptsize Player 1 Infosets \\
\scriptsize{\textcolor{player2_2}{\DiamondSolid} \textcolor{player2_1}{\DiamondSolid}} & \scriptsize Player 2 Infosets \\
\scriptsize{\textcolor{terminal}{\TriangleUp}} & \scriptsize Terminal Nodes
\end{tabular}}
\end{flushright}
\end{subfigure}

    \caption{Sample empirical game tree with initial policy set at start of TE-PSRO}
    \label{fig:empirical_game_iter0}
\end{figure}

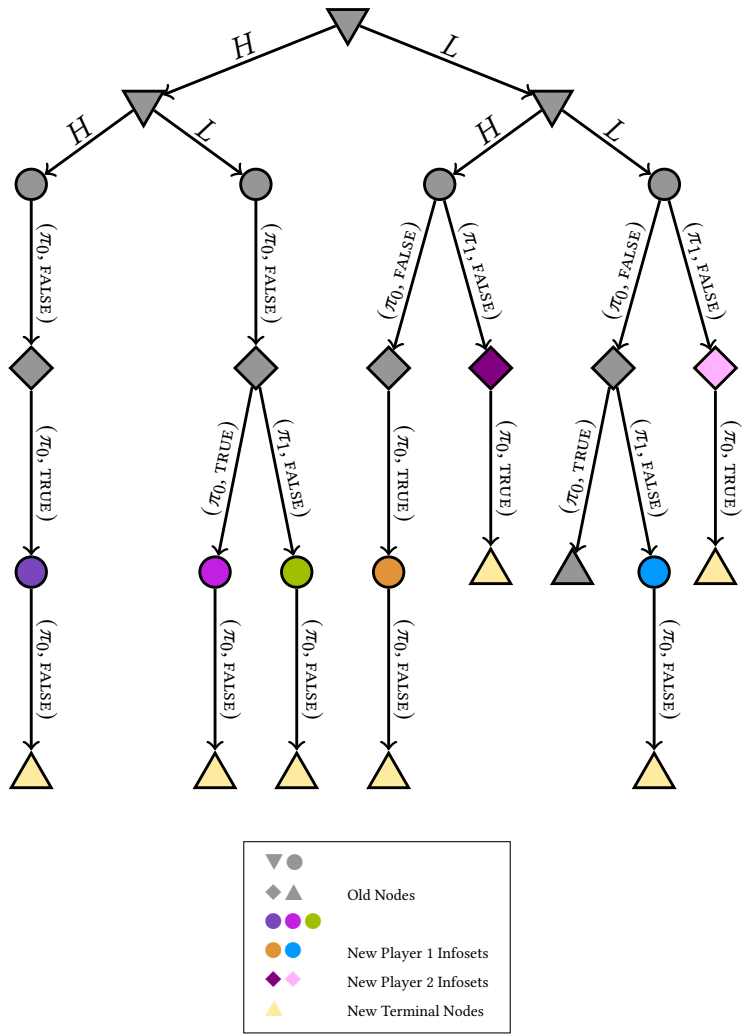
\begin{figure}[ht!]
\centering
\begin{subfigure}[t]{\textwidth}
\centering
\resizebox{0.55\textwidth}{!}{%
\begin{tikzpicture}[thick,
    level 1/.style = {level distance = 8mm, sibling distance = 40mm},
    level 2/.style = {level distance = 8mm, sibling distance = 22mm},
    level 3/.style = {level distance = 18mm, sibling distance = 10mm},
    level 4/.style = {level distance = 20mm, sibling distance = 8mm},
    engine/.style = {inner sep = 1pt, above}]
    \node [draw, black, fill={oldplayer}, regular polygon, regular polygon sides=3, rotate=180, inner sep=0.08cm] {}
    [black, ->]

    child {node [draw, black, fill={oldplayer}, regular polygon, regular polygon sides=3, rotate=180, inner sep=0.08cm] (1) {} 
      child {node [draw, black, fill={oldplayer}, circle] (2) {} 
        child {node [draw, black, fill={oldplayer}, diamond] (3) {}
          child {node [draw, black, fill={player1_5}, circle] (4) {}
            child {node [draw, black, fill={terminal}, regular polygon, regular polygon sides=3, inner sep=0.08cm] (6) {}
            edge from parent node[engine, sloped] {{\scriptsize $(\pi_0, \F)$}}}
          edge from parent node[engine, sloped] {{\scriptsize $(\pi_0, \T)$}}}
        edge from parent node[engine, sloped] {{\scriptsize $(\pi_0, \F)$}}}
      edge from parent node[engine, sloped] {$H$}}
      child {node [draw, black, fill={oldplayer}, circle] (5) {} 
        child {node [draw, black, fill={oldplayer}, diamond] (3) {}
          child {node [draw, black, fill={player1_6}, circle] (4) {}
            child {node [draw, black, fill={terminal}, regular polygon, regular polygon sides=3, inner sep=0.08cm] (6) {}
            edge from parent node[engine, sloped] {{\scriptsize $(\pi_0, \F)$}}}
          edge from parent node[engine, sloped] {{\scriptsize $(\pi_0, \T)$}}}
          child {node [draw, black, fill={player1_4}, circle] (4) {}
            child {node [draw, black, fill={terminal}, regular polygon, regular polygon sides=3, inner sep=0.08cm] (6) {}
            edge from parent node[engine, sloped] {{\scriptsize $(\pi_0, \F)$}}}
          edge from parent node[engine, sloped] {{\scriptsize $(\pi_1, \F)$}}}
        edge from parent node[engine, sloped] {{\scriptsize $(\pi_0, \F)$}}}
      edge from parent node[engine, sloped] {$L$}}
    edge from parent node[engine, sloped] {$H$}}
    child {node [draw, black, fill={oldplayer}, regular polygon, regular polygon sides=3, rotate=180, inner sep=0.08cm] (1) {} 
      child {node [draw, black, fill={oldplayer}, circle] (2) {} 
        child {node [draw, black, fill={oldplayer}, diamond] (3) {}
          child {node [draw, black, fill={player1_3}, circle] (4) {}
            child {node [draw, black, fill={terminal}, regular polygon, regular polygon sides=3, inner sep=0.08cm] (6) {}
            edge from parent node[engine, sloped] {{\scriptsize $(\pi_0, \F)$}}}
          edge from parent node[engine, sloped] {{\scriptsize $(\pi_0, \T)$}}}
        edge from parent node[engine, sloped] {{\scriptsize $(\pi_0, \F)$}}}
        child {node [draw, black, fill={player2_3}, diamond] (7) {}
          child {node [draw, black, fill={terminal}, regular polygon, regular polygon sides=3, inner sep=0.08cm] (8) {}
          edge from parent node[engine, sloped] {{\scriptsize $(\pi_0, \T)$}}}
        edge from parent node[engine, sloped] {{\scriptsize $(\pi_1, \F)$}}}
      edge from parent node[engine, sloped] {$H$}}
      child {node [draw, black, fill={oldplayer}, circle] (5) {} 
        child {node [draw, black, fill={oldplayer}, diamond] (3) {}
          child {node [draw, black, fill={oldplayer}, regular polygon, regular polygon sides=3, inner sep=0.08cm] (4) {}
          edge from parent node[engine, sloped] {{\scriptsize $(\pi_0, \T)$}}}
          child {node [draw, black, fill={player1_2}, circle] (4) {}
            child {node [draw, black, fill={terminal}, regular polygon, regular polygon sides=3, inner sep=0.08cm] (6) {}
            edge from parent node[engine, sloped] {{\scriptsize $(\pi_0, \F)$}}}
          edge from parent node[engine, sloped] {{\scriptsize $(\pi_1, \F)$}}}
        edge from parent node[engine, sloped] {{\scriptsize $(\pi_0, \F)$}}}
        child {node [draw, black, fill={player2_4}, diamond] (8) {}
          child {node [draw, black, fill={terminal}, regular polygon, regular polygon sides=3, inner sep=0.08cm] (9) {}
          edge from parent node[engine, sloped] {{\scriptsize $(\pi_0, \T)$}}}
        edge from parent node[engine, sloped] {{\scriptsize $(\pi_1, \F)$}}}
      edge from parent node[engine, sloped] {$L$}}
    edge from parent node[engine, sloped] {$L$}
    };

\end{tikzpicture}%
}
\end{subfigure}
\bigskip

\begin{subfigure}[t]{.35\textwidth}
\centering
    \fbox{\begin{tabular}{ll}
\scriptsize{\textcolor{oldplayer}{\TriangleDown} \textcolor{oldplayer}{\CircleSolid}} \\ \scriptsize{\textcolor{oldplayer}{\DiamondSolid} \textcolor{oldplayer}{\TriangleUp}} & \scriptsize Old Nodes \\
\scriptsize{\textcolor{player1_5}{\CircleSolid} \textcolor{player1_6}{\CircleSolid} \textcolor{player1_4}{\CircleSolid}} \\ \scriptsize{\textcolor{player1_3}{\CircleSolid} \textcolor{player1_2}{\CircleSolid}} & \scriptsize New Player 1 Infosets \\
\scriptsize{\textcolor{player2_3}{\DiamondSolid} \textcolor{player2_4}{\DiamondSolid}} & \scriptsize New Player 2 Infosets \\
\scriptsize{\textcolor{terminal}{\TriangleUp}} & \scriptsize New Terminal Nodes
\end{tabular}}
\end{subfigure}
    \caption{Empirical game tree after iteration 1 of TE-PSRO for $M = 1$.}
    \label{fig:empirical_game_iter1}
\end{figure}

Let player 1's new best response be $(\pi_1, \F)$ and player 2's new best response be $(\pi_1, \F)$. Since $M = 1$, suppose player 1's randomly chosen information set is $(L,)$ and player 2's chosen information set is $(L, (\pi_0, \F))$. All possible strategy profiles for the empirical game tree given the new $\hat{\cI_1} \cup \hat{\cI_2}$ are then passed into the simulator one at a time and sampled with a fixed budget. The simulator then returns a corresponding history of actions \textit{through the empirical game tree} representing a new path and sampled reward vector. All possible new paths are then added to the empirical game tree, with each (possibly new) terminal node containing a weighted average of all payoff samples associated with the leaf's history known as the \term{leaf utility}. Figure~\ref{fig:empirical_game_iter1} illustrates the resulting new empirical game tree, where new information sets are brightly colored and all old nodes included in the initial empirical game tree are gray. Additional figures that visualize the addition of new best responses to $\hat{G}$ for $M = 1$ during iterations 2 and 3 of TE-PSRO are included in the appendices.

In iteration 2 of TE-PSRO, suppose that player 1's best response $(\pi_2, \F)$ is added to the randomly chosen information set $(H, L,(\pi_0, \F),(\pi_0, \T))$
highlighted in bright pink in Figure~\ref{fig:empirical_game_iter1} and that player 2's best response $(\pi_2, \T)$ is added to the randomly chosen information set $(H, (\pi_0, \F))$ highlighted in both gray and dark pink in Figure~\ref{fig:empirical_game_iter1}. The resulting expanded game tree with all possible new paths is depicted in Figure~\ref{fig:empirical_game_iter2}. Notice that in this figure, the terminal nodes from the previous iteration's empirical game tree are replaced with new information sets yet again, mostly belonging to player 2. Finally, in iteration 3 of TE-PSRO, suppose that player 1's best response policy $(\pi_3, \T)$ gets added to the information set $(H,)$ early on in the tree and player 2's best response policy $(\pi_3, \F)$ gets added to the information set $(H, (\pi_1, \F))$. The resulting expanded empirical game tree is depicted in Figure~\ref{fig:empirical_game_iter3}.

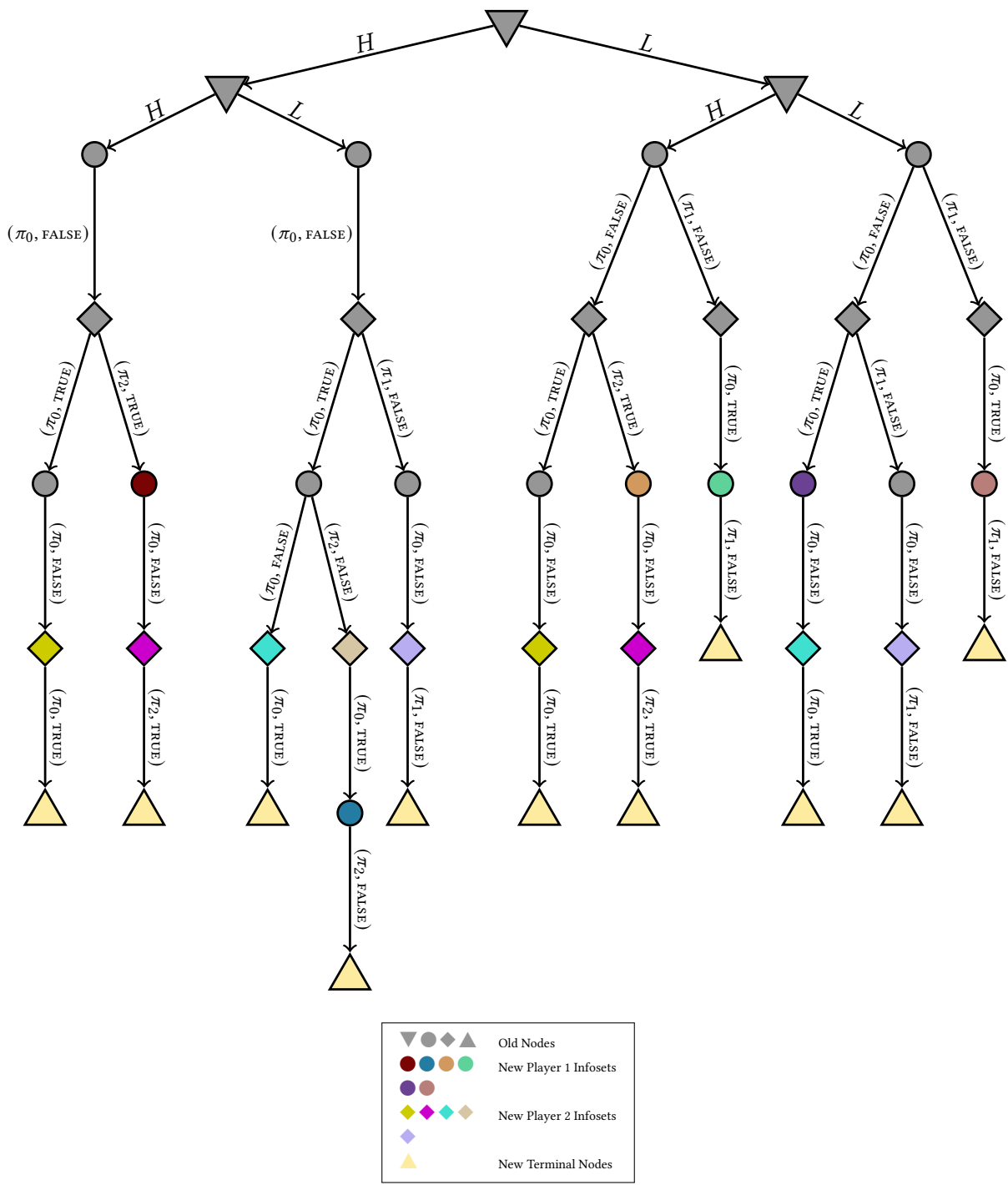
\begin{figure}[H]
\centering
\begin{subfigure}[b]{\textwidth}
\centering
\resizebox{0.9\textwidth}{!}{%
\begin{tikzpicture}[thick,
    level 1/.style = {level distance = 8mm, sibling distance = 68mm},
    level 2/.style = {level distance = 8mm, sibling distance = 32mm},
    level 3/.style = {level distance = 20mm, sibling distance = 16mm},
    level 4/.style = {level distance = 20mm, sibling distance = 12mm},
    level 5/.style = {level distance = 20mm, sibling distance = 10mm},
    engine/.style = {inner sep = 1pt, above}]
    \node [draw, black, fill={oldplayer}, regular polygon, regular polygon sides=3, rotate=180, inner sep=0.1cm] {}
    [black, ->]

    child {node [draw, black, fill={oldplayer}, regular polygon, regular polygon sides=3, rotate=180, inner sep=0.1cm] (1) {} 
      child {node [draw, black, fill={oldplayer}, circle] (2) {} 
        child {node [draw, black, fill={oldplayer}, diamond] (3) {}
          child {node [draw, black, fill={oldplayer}, circle] (4) {}
            child {node [draw, black, fill={player2_8}, diamond] (6) {}
              child {node [draw, black, fill={terminal}, regular polygon, regular polygon sides=3, inner sep=0.1cm] (20) {}
              edge from parent node[engine, sloped] {{\scriptsize $(\pi_0, \T)$}}}
            edge from parent node[engine, sloped] {{\scriptsize $(\pi_0, \F)$}}}
          edge from parent node[engine, sloped] {{\scriptsize $(\pi_0, \T)$}}}
          child {node [draw, black, fill={player1_7}, circle] (10) {}
            child {node [draw, black, fill={player2_7}, diamond] (11) {}
              child {node [draw, black, fill={terminal}, regular polygon, regular polygon sides=3, inner sep=0.1cm] (19) {}
              edge from parent node[engine, sloped] {{\scriptsize $(\pi_2, \T)$}}}
            edge from parent node[engine, sloped] {{\scriptsize $(\pi_0, \F)$}}}
          edge from parent node[engine, sloped] {{\scriptsize $(\pi_2, \T)$}}}
        edge from parent node[engine, left] {{\scriptsize $(\pi_0, \F)$}}}
      edge from parent node[engine, sloped] {$H$}}
      child {node [draw, black, fill={oldplayer}, circle] (5) {} 
        child {node [draw, black, fill={oldplayer}, diamond] (3) {}
          child {node [draw, black, fill={oldplayer}, circle] (4) {}
            child {node [draw, black, fill={player2_10}, diamond] (6) {}
              child {node [draw, black, fill={terminal}, regular polygon, regular polygon sides=3, inner sep=0.1cm] (21) {}
              edge from parent node[engine, sloped] {{\scriptsize $(\pi_0, \T)$}}}
            edge from parent node[engine, sloped] {{\scriptsize $(\pi_0, \F)$}}}
            child {node [draw, black, fill={player2_6}, diamond] (8) {}
              child {node [draw, black, fill={player1_8}, circle] (12) {}
                child {node [draw, black, fill={terminal}, regular polygon, regular polygon sides=3, inner sep=0.1cm] (13) {}
                edge from parent node[engine, sloped] {{\scriptsize $(\pi_2, \F)$}}}
              edge from parent node[engine, sloped] {{\scriptsize $(\pi_0, \T)$}}}
            edge from parent node[engine, sloped] {{\scriptsize $(\pi_2, \F)$}}}
          edge from parent node[engine, sloped] {{\scriptsize $(\pi_0, \T)$}}}
          child {node [draw, black, fill={oldplayer}, circle] (4) {}
            child {node [draw, black, fill={player2_9}, diamond] (6) {}
              child {node [draw, black, fill={terminal}, regular polygon, regular polygon sides=3, inner sep=0.1cm] (21) {}
              edge from parent node[engine, sloped] {{\scriptsize $(\pi_1, \F)$}}}
            edge from parent node[engine, sloped] {{\scriptsize $(\pi_0, \F)$}}}
          edge from parent node[engine, sloped] {{\scriptsize $(\pi_1, \F)$}}}
        edge from parent node[engine, left] {{\scriptsize $(\pi_0, \F)$}}}
      edge from parent node[engine, sloped] {$L$}}
    edge from parent node[engine, sloped] {$H$}}
    child {node [draw, black, fill={oldplayer}, regular polygon, regular polygon sides=3, rotate=180, inner sep=0.1cm] (1) {} 
      child {node [draw, black, fill={oldplayer}, circle] (2) {} 
        child {node [draw, black, fill={oldplayer}, diamond] (3) {}
          child {node [draw, black, fill={oldplayer}, circle] (4) {}
            child {node [draw, black, fill={player2_8}, diamond] (6) {}
              child {node [draw, black, fill={terminal}, regular polygon, regular polygon sides=3, inner sep=0.1cm] (20) {}
              edge from parent node[engine, sloped] {{\scriptsize $(\pi_0, \T)$}}}
            edge from parent node[engine, sloped] {{\scriptsize $(\pi_0, \F)$}}}
          edge from parent node[engine, sloped] {{\scriptsize $(\pi_0, \T)$}}}
          child {node [draw, black, fill={player1_9}, circle] (12) {}
            child {node [draw, black, fill={player2_7}, diamond] (13) {}
              child {node [draw, black, fill={terminal}, regular polygon, regular polygon sides=3, inner sep=0.1cm] (19) {}
              edge from parent node[engine, sloped] {{\scriptsize $(\pi_2, \T)$}}}
            edge from parent node[engine, sloped] {{\scriptsize $(\pi_0, \F)$}}}
          edge from parent node[engine, sloped] {{\scriptsize $(\pi_2, \T)$}}}
        edge from parent node[engine, sloped] {{\scriptsize $(\pi_0, \F)$}}}
        child {node [draw, black, fill={oldplayer}, diamond] (7) {}
          child {node [draw, black, fill={player1_10}, circle] (8) {}
            child {node [draw, black, fill={terminal}, regular polygon, regular polygon sides=3, inner sep=0.1cm] (14) {}
            edge from parent node[engine, sloped] {{\scriptsize $(\pi_1, \F)$}}}
          edge from parent node[engine, sloped] {{\scriptsize $(\pi_0, \T)$}}}
        edge from parent node[engine, sloped] {{\scriptsize $(\pi_1, \F)$}}}
      edge from parent node[engine, sloped] {$H$}}
      child {node [draw, black, fill={oldplayer}, circle] (5) {} 
        child {node [draw, black, fill={oldplayer}, diamond] (3) {}
          child {node [draw, black, fill={player1_11}, circle] (4) {}
            child {node [draw, black, fill={player2_10}, diamond] (6) {}
              child {node [draw, black, fill={terminal}, regular polygon, regular polygon sides=3, inner sep=0.1cm] (21) {}
              edge from parent node[engine, sloped] {{\scriptsize $(\pi_0, \T)$}}}
            edge from parent node[engine, sloped] {{\scriptsize $(\pi_0, \F)$}}}
          edge from parent node[engine, sloped] {{\scriptsize $(\pi_0, \T)$}}}
          child {node [draw, black, fill={oldplayer}, circle] (4) {}
            child {node [draw, black, fill={player2_9}, diamond] (6) {}
              child {node [draw, black, fill={terminal}, regular polygon, regular polygon sides=3, inner sep=0.1cm] (21) {}
              edge from parent node[engine, sloped] {{\scriptsize $(\pi_1, \F)$}}}
            edge from parent node[engine, sloped] {{\scriptsize $(\pi_0, \F)$}}}
          edge from parent node[engine, sloped] {{\scriptsize $(\pi_1, \F)$}}}
        edge from parent node[engine, sloped] {{\scriptsize $(\pi_0, \F)$}}}
        child {node [draw, black, fill={oldplayer}, diamond] (8) {}
          child {node [draw, black, fill={player1_12}, circle] (9) {}
            child {node [draw, black, fill={terminal}, regular polygon, regular polygon sides=3, inner sep=0.1cm] (16) {}
            edge from parent node[engine, sloped] {{\scriptsize $(\pi_1, \F)$}}}
          edge from parent node[engine, sloped] {{\scriptsize $(\pi_0, \T)$}}}
        edge from parent node[engine, sloped] {{\scriptsize $(\pi_1, \F)$}}}
      edge from parent node[engine, sloped] {$L$}}
    edge from parent node[engine, sloped] {$L$}
    };

\end{tikzpicture}%
}
\vspace{1.5em}
\end{subfigure}
\begin{subfigure}[b]{0.33\textwidth}
\centering
    \fbox{\begin{tabular}{ll}
\footnotesize{\textcolor{oldplayer}{\TriangleDown} \textcolor{oldplayer}{\CircleSolid} \textcolor{oldplayer}{\DiamondSolid} \textcolor{oldplayer}{\TriangleUp}} & \scriptsize Old Nodes \\
\footnotesize{\textcolor{player1_7}{\CircleSolid} \textcolor{player1_8}{\CircleSolid} \textcolor{player1_9}{\CircleSolid} \textcolor{player1_10}{\CircleSolid}} & \scriptsize New Player 1 Infosets \\
\footnotesize{\textcolor{player1_11}{\CircleSolid} \textcolor{player1_12}{\CircleSolid}} & \scriptsize \\
\footnotesize{\textcolor{player2_8}{\DiamondSolid} \textcolor{player2_7}{\DiamondSolid} \textcolor{player2_10}{\DiamondSolid} \textcolor{player2_6}{\DiamondSolid}} & \scriptsize New Player 2 Infosets \\
\footnotesize{\textcolor{player2_9}{\DiamondSolid}} & \scriptsize \\
\footnotesize{\textcolor{terminal}{\TriangleUp}} & \scriptsize New Terminal Nodes
\end{tabular}}
\end{subfigure}
    \caption{Empirical game tree after iteration 2 of TE-PSRO where a new BR policy has been added to a single information state in $\hat{G}$ for each player; here, the player 1 best response $(\pi_2, \F)$ has been added to the player 1 infoset $(H, L, (\pi_0, \F), (\pi_0, \T))$ while the player 2 best response $(\pi_2, \T)$ has been added to the player 2 infoset $(H, (\pi_0, \F))$.}
    \label{fig:empirical_game_iter2}
\end{figure}

\newpage

Iteration 3

\begin{figure}[H]
\centering
\begin{subfigure}[b]{\textwidth}
\centering
\resizebox{\textwidth}{!}{%
\begin{tikzpicture}[thick,
    level 1/.style = {level distance = 8mm, sibling distance = 80mm},
    level 2/.style = {level distance = 8mm, sibling distance = 40mm},
    level 3/.style = {level distance = 18mm, sibling distance = 22mm},
    level 4/.style = {level distance = 18mm, sibling distance = 12mm},
    level 5/.style = {level distance = 18mm, sibling distance = 10mm},
    level 6/.style = {level distance = 18mm, sibling distance = 10mm},
    engine/.style = {inner sep = 1pt, above}]
    \node [draw, black, fill={oldplayer}, regular polygon, regular polygon sides=3, rotate=180, inner sep=0.1cm] {}
    [black, ->]

    child {node [draw, black, fill={oldplayer}, regular polygon, regular polygon sides=3, rotate=180, inner sep=0.1cm] (1) {} 
      child {node [draw, black, fill={oldplayer}, circle] (2) {} 
        child {node [draw, black, fill={oldplayer}, diamond] (3) {}
          child {node [draw, black, fill={oldplayer}, circle] (4) {}
            child {node [draw, black, fill={oldplayer}, diamond] (6) {}
              child {node [draw, black, fill={oldplayer}, regular polygon, regular polygon sides=3, inner sep=0.1cm] (20) {}
              edge from parent node[engine, sloped] {{\scriptsize $(\pi_0, \T)$}}}
              child {node [draw, black, fill={player1_17}, circle] (26) {}
                child {node [draw, black, fill={terminal}, regular polygon, regular polygon sides=3, inner sep=0.1cm] (27) {}
                edge from parent node[engine, sloped] {{\scriptsize $(\pi_0, \F)$}}}
              edge from parent node[engine, sloped] {{\scriptsize $(\pi_2, \T)$}}}
            edge from parent node[engine, sloped] {{\scriptsize $(\pi_0, \F)$}}}
          edge from parent node[engine, sloped] {{\scriptsize $(\pi_0, \T)$}}}
          child {node [draw, black, fill={oldplayer}, circle] (10) {}
            child {node [draw, black, fill={oldplayer}, diamond] (11) {}
              child {node [draw, black, fill={player1_15}, circle] (23) {}
                child {node [draw, black, fill={terminal}, regular polygon, regular polygon sides=3, inner sep=0.1cm] (19) {}
                edge from parent node[engine, sloped] {{\scriptsize $(\pi_0, \F)$}}}
              edge from parent node[engine, sloped] {{\scriptsize $(\pi_2, \T)$}}}
            edge from parent node[engine, sloped] {{\scriptsize $(\pi_0, \F)$}}}
          edge from parent node[engine, sloped] {{\scriptsize $(\pi_2, \T)$}}}
        edge from parent node[engine, sloped] {{\scriptsize $(\pi_0, \F)$}}}
        child {node [draw, black, fill={player2_4}, diamond] (8) {}
          child {node [draw, black, fill={terminal}, regular polygon, regular polygon sides=3, inner sep=0.1cm] (31) {}
          edge from parent node[engine, sloped] {{\scriptsize $(\pi_0, \T)$}}}
        edge from parent node[engine, sloped] {{\scriptsize $(\pi_3, \T)$}}}
      edge from parent node[engine, sloped] {$H$}}
      child {node [draw, black, fill={oldplayer}, circle] (5) {} 
        child {node [draw, black, fill={oldplayer}, diamond] (3) {}
          child {node [draw, black, fill={oldplayer}, circle] (4) {}
            child {node [draw, black, fill={oldplayer}, diamond] (6) {}
              child {node [draw, black, fill={player1_2}, circle] (24) {}
                child {node [draw, black, fill={terminal}, regular polygon, regular polygon sides=3, inner sep=0.1cm] (21) {}
                edge from parent node[engine, sloped] {{\scriptsize $(\pi_0, \F)$}}}
              edge from parent node[engine, sloped] {{\scriptsize $(\pi_0, \T)$}}}
            edge from parent node[engine, sloped] {{\scriptsize $(\pi_0, \F)$}}}
            child {node [draw, black, fill={oldplayer}, diamond] (8) {}
              child {node [draw, black, fill={oldplayer}, circle] (12) {}
                child {node [draw, black, fill={oldplayer}, regular polygon, regular polygon sides=3, inner sep=0.1cm] (13) {}
                edge from parent node[engine, sloped] {{\scriptsize $(\pi_2, \F)$}}}
              edge from parent node[engine, sloped] {{\scriptsize $(\pi_0, \T)$}}}
            edge from parent node[engine, sloped] {{\scriptsize $(\pi_2, \F)$}}}
          edge from parent node[engine, sloped] {{\scriptsize $(\pi_0, \T)$}}}
          child {node [draw, black, fill={oldplayer}, circle] (4) {}
            child {node [draw, black, fill={oldplayer}, diamond] (6) {}
              child {node [draw, black, fill={player1_14}, circle] (22) {}
                child {node [draw, black, fill={terminal}, regular polygon, regular polygon sides=3, inner sep=0.1cm] (21) {}
                edge from parent node[engine, sloped] {{\scriptsize $(\pi_0, \F)$}}}
              edge from parent node[engine, sloped] {{\scriptsize $(\pi_1, \F)$}}}
            edge from parent node[engine, sloped] {{\scriptsize $(\pi_0, \F)$}}}
          edge from parent node[engine, sloped] {{\scriptsize $(\pi_1, \F)$}}}
        edge from parent node[engine, sloped] {{\scriptsize $(\pi_0, \F)$}}}
        child {node [draw, black, fill={player2_1}, diamond] (10) {}
          child {node [draw, black, fill={terminal}, regular polygon, regular polygon sides=3, inner sep=0.1cm] (32) {}
          edge from parent node[engine, sloped] {{\scriptsize $(\pi_0, \T)$}}}
        edge from parent node[engine, sloped] {{\scriptsize $(\pi_3, \T)$}}}
      edge from parent node[engine, sloped] {$L$}}
    edge from parent node[engine, sloped] {$H$}}
    child {node [draw, black, fill={oldplayer}, regular polygon, regular polygon sides=3, rotate=180, inner sep=0.1cm] (1) {} 
      child {node [draw, black, fill={oldplayer}, circle] (2) {} 
        child {node [draw, black, fill={oldplayer}, diamond] (3) {}
          child {node [draw, black, fill={oldplayer}, circle] (4) {}
            child {node [draw, black, fill={oldplayer}, diamond] (6) {}
              child {node [draw, black, fill={oldplayer}, regular polygon, regular polygon sides=3, inner sep=0.1cm] (20) {}
              edge from parent node[engine, sloped] {{\scriptsize $(\pi_0, \T)$}}}
              child {node [draw, black, fill={player1_18}, circle] (27) {}
                child {node [draw, black, fill={terminal}, regular polygon, regular polygon sides=3, inner sep=0.1cm] (28) {}
                edge from parent node[engine, sloped] {{\scriptsize $(\pi_0, \F)$}}}
              edge from parent node[engine, sloped] {{\scriptsize $(\pi_2, \T)$}}}
            edge from parent node[engine, sloped] {{\scriptsize $(\pi_0, \F)$}}}
          edge from parent node[engine, sloped] {{\scriptsize $(\pi_0, \T)$}}}
          child {node [draw, black, fill={oldplayer}, circle] (12) {}
            child {node [draw, black, fill={oldplayer}, diamond] (13) {}
              child {node [draw, black, fill={oldplayer}, regular polygon, regular polygon sides=3, inner sep=0.1cm] (19) {}
              edge from parent node[engine, sloped] {{\scriptsize $(\pi_2, \T)$}}}
            edge from parent node[engine, sloped] {{\scriptsize $(\pi_0, \F)$}}}
          edge from parent node[engine, sloped] {{\scriptsize $(\pi_2, \T)$}}}
        edge from parent node[engine, sloped] {{\scriptsize $(\pi_0, \F)$}}}
        child {node [draw, black, fill={oldplayer}, diamond] (7) {}
          child {node [draw, black, fill={oldplayer}, circle] (8) {}
            child {node [draw, black, fill={player2_12}, diamond] (33) {}
              child {node [draw, black, fill={terminal}, regular polygon, regular polygon sides=3, inner sep=0.1cm] (14) {}
              edge from parent node[engine, sloped] {{\scriptsize $(\pi_0, \T)$}}}
            edge from parent node[engine, sloped] {{\scriptsize $(\pi_1, \F)$}}}
          edge from parent node[engine, sloped] {{\scriptsize $(\pi_0, \T)$}}}
          child {node [draw, black, fill={player1_19}, circle] (29) {}
            child {node [draw, black, fill={terminal}, regular polygon, regular polygon sides=3, inner sep=0.1cm] (30) {}
            edge from parent node[engine, sloped] {{\scriptsize $(\pi_1, \F)$}}}
          edge from parent node[engine, sloped] {{\scriptsize $(\pi_3, \F)$}}}
        edge from parent node[engine, sloped] {{\scriptsize $(\pi_1, \F)$}}}
      edge from parent node[engine, sloped] {$H$}}
      child {node [draw, black, fill={oldplayer}, circle] (5) {} 
        child {node [draw, black, fill={oldplayer}, diamond] (3) {}
          child {node [draw, black, fill={oldplayer}, circle] (4) {}
            child {node [draw, black, fill={oldplayer}, diamond] (6) {}
              child {node [draw, black, fill={player1_20}, circle] (30) {}
                child {node [draw, black, fill={terminal}, regular polygon, regular polygon sides=3, inner sep=0.1cm] (21) {}
                edge from parent node[engine, sloped] {{\scriptsize $(\pi_0, \F)$}}}
              edge from parent node[engine, sloped] {{\scriptsize $(\pi_0, \T)$}}}
            edge from parent node[engine, sloped] {{\scriptsize $(\pi_0, \F)$}}}
          edge from parent node[engine, sloped] {{\scriptsize $(\pi_0, \T)$}}}
          child {node [draw, black, fill={oldplayer}, circle] (4) {}
            child {node [draw, black, fill={oldplayer}, diamond] (6) {}
              child {node [draw, black, fill={player1_1}, circle] (27) {}
                child {node [draw, black, fill={terminal}, regular polygon, regular polygon sides=3, inner sep=0.1cm] (21) {}
                edge from parent node[engine, sloped] {{\scriptsize $(\pi_0, \F)$}}}
              edge from parent node[engine, sloped] {{\scriptsize $(\pi_1, \F)$}}}
            edge from parent node[engine, sloped] {{\scriptsize $(\pi_0, \F)$}}}
          edge from parent node[engine, sloped] {{\scriptsize $(\pi_1, \F)$}}}
        edge from parent node[engine, sloped] {{\scriptsize $(\pi_0, \F)$}}}
        child {node [draw, black, fill={oldplayer}, diamond] (8) {}
          child {node [draw, black, fill={oldplayer}, circle] (9) {}
            child {node [draw, black, fill={player2_11}, diamond] (31) {}
              child {node [draw, black, fill={terminal}, regular polygon, regular polygon sides=3, inner sep=0.1cm] (16) {}
              edge from parent node[engine, sloped] {{\scriptsize $(\pi_0, \T)$}}}
            edge from parent node[engine, sloped] {{\scriptsize $(\pi_1, \F)$}}}
          edge from parent node[engine, sloped] {{\scriptsize $(\pi_0, \T)$}}}
        edge from parent node[engine, sloped] {{\scriptsize $(\pi_1, \F)$}}}
      edge from parent node[engine, sloped] {$L$}}
    edge from parent node[engine, sloped] {$L$}
    };

\end{tikzpicture}%
}
\vspace{1.5em}
\end{subfigure}
\begin{subfigure}[b]{0.4\textwidth}
\centering
    \fbox{\begin{tabular}{ll}
\footnotesize{\textcolor{oldplayer}{\TriangleDown} \textcolor{oldplayer}{\CircleSolid} \textcolor{oldplayer}{\DiamondSolid} \textcolor{oldplayer}{\TriangleUp}} & \scriptsize Old Nodes \\
\footnotesize{\textcolor{player1_17}{\CircleSolid} \textcolor{player1_15}{\CircleSolid} \textcolor{player1_2}{\CircleSolid} \textcolor{player1_14}{\CircleSolid}} & \scriptsize New Player 1 Infosets \\
\footnotesize{\textcolor{player1_18}{\CircleSolid} \textcolor{player1_19}{\CircleSolid} \textcolor{player1_20}{\CircleSolid} \textcolor{player1_1}{\CircleSolid}} &  \\
\footnotesize{\textcolor{player2_4}{\DiamondSolid} \textcolor{player2_1}{\DiamondSolid} \textcolor{player2_12}{\DiamondSolid} \textcolor{player2_11}{\DiamondSolid}} & \scriptsize New Player 2 Infosets \\
\footnotesize{\textcolor{terminal}{\TriangleUp}} & \scriptsize New Terminal Nodes
\end{tabular}}
\end{subfigure}
    \caption{Empirical game tree after iteration 3 of TE-PSRO where a new BR policy has been added to a single information state in $\hat{G}$ for each player; here, the player 1 best response $(\pi_3, \T)$ has been added to the player 1 infoset $(H,)$ while the player 2 best response $(\pi_3, \F)$ has been added to the player 2 infoset $(H, (\pi_1, \F))$.}
    \label{fig:empirical_game_iter3}
\end{figure}
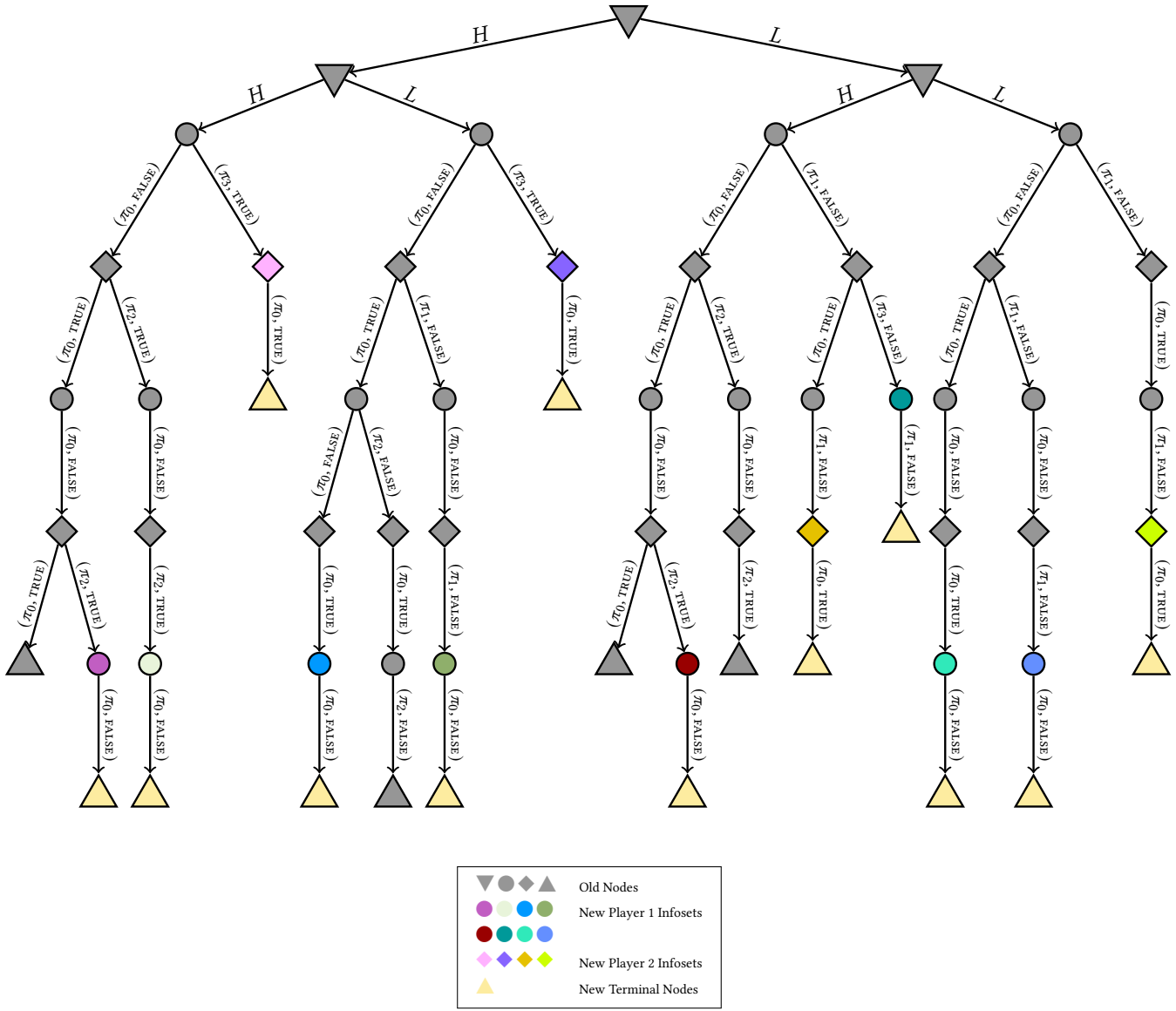

\newpage
\subsection{Game Tree Expansion}\label{sec:expand_game_tree}

It is important to note that several of what were originally terminal nodes in $\hat{G}$ (regardless of how many rounds of negotiation have passed) have been replaced with new information sets belonging to player 1, along with a ``default policy" edge leading to a brand new terminal node (represented with yellow triangles). The default policy edge is labeled on the assumption that the player has continued with the same policy he committed to in his previous information set in the sequence, since a policy $\pi_x$ specifies what that player should do in every possible game state of the true game. When adding new simulation data to the empirical game tree, we are adding a new history of policy-and-signal-reveal pairs to the tree along with its corresponding sampled reward vector. Although simple in principle, the fact that gameplay can end after anywhere from 2 to $2T$ player turns have been completed at different points in the game tree complicates the process of expanding the tree with new simulation data in a manner that \citet{konicki22} did not address since their considerations were restricted to games whose terminal nodes all had the same history length. We demonstrate how a new observation sequence and sampled reward outputted by the simulator might be added to $\hat{G}$ for the bargaining game via three scenarios.

In Scenario 1, the tree in Figure~\ref{fig:scenario1_a} is expanded to include a completely new terminal node whose history does not overlap with any other present terminal nodes. The lack of overlap between the two terminal nodes' respective histories after the first decision node make depicted in Figure~\ref{fig:scenario1_b} meant the second terminal node could be added without overwriting the first terminal node's history, position in the tree, or leaf utility. 
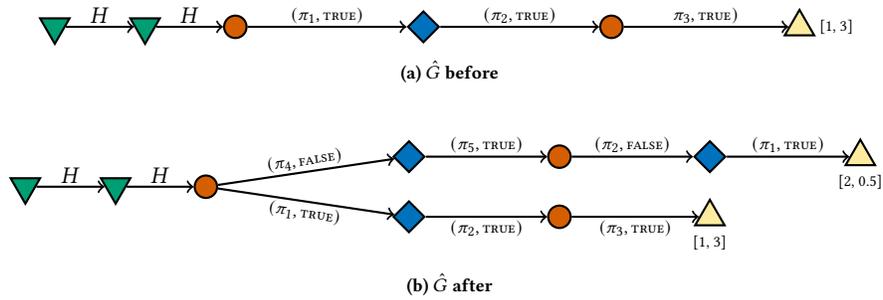
\begin{figure}[ht!]
\centering
\begin{subfigure}[b]{\textwidth}
\centering
\begin{tikzpicture}[thick,
    level 1/.style = {level distance = 12mm, sibling distance = 80mm},
    level 2/.style = {level distance = 12mm, sibling distance = 38mm},
    level 3/.style = {level distance = 25mm, sibling distance = 13mm},
    level 4/.style = {level distance = 25mm, sibling distance = 8mm},
    level 5/.style = {level distance = 25mm, sibling distance = 8mm},
    engine/.style = {inner sep = 1pt, above},
    grow=right]
    \node [draw, black, fill={chance}, regular polygon, regular polygon sides=3, rotate=180, inner sep=0.08cm] {}
    [black, ->]

    child {node [draw, black, fill={chance}, regular polygon, regular polygon sides=3, rotate=180, inner sep=0.08cm] (1) {} 
      child {node [draw, black, fill={player1}, circle] (2) {} 
        child {node [draw, black, fill={player2}, diamond] (3) {}
          child {node [draw, black, fill={player1}, circle] (4) {}
            child {node [draw, black, fill={terminal},  regular polygon, regular polygon sides=3, inner sep=0.08cm, label=right:{{\scriptsize[1, 3]}}] (5) {}
            edge from parent node[engine, sloped] {{\scriptsize $\mathbf{\pi}_3, \T)$}}}
          edge from parent node[engine, sloped] {{\scriptsize $(\mathbf{\pi}_2, \T)$}}}
        edge from parent node[engine, sloped] {{\scriptsize $(\mathbf{\pi}_1, \T)$}}}
      edge from parent node[engine, sloped] {$H$}}
    edge from parent node[engine, sloped] {$H$}
    };

\end{tikzpicture}
\caption{$\hat{G}$ before }\label{fig:scenario1_a}
\vspace{2em}
\end{subfigure} \\
\begin{subfigure}[b]{\textwidth}
\centering

\begin{tikzpicture}[thick,
    level 1/.style = {level distance = 12mm, sibling distance = 80mm},
    level 2/.style = {level distance = 12mm, sibling distance = 38mm},
    level 3/.style = {level distance = 27mm, sibling distance = 8mm},
    level 4/.style = {level distance = 20mm, sibling distance = 8mm},
    level 5/.style = {level distance = 20mm, sibling distance = 8mm},
    engine/.style = {inner sep = 1pt, above},
    grow=right]
    \node [draw, black, fill={chance}, regular polygon, regular polygon sides=3, rotate=180, inner sep=0.08cm] {}
    [black, ->]

    child {node [draw, black, fill={chance}, regular polygon, regular polygon sides=3, rotate=180, inner sep=0.08cm] (1) {} 
      child {node [draw, black, fill={player1}, circle] (2) {} 
        child {node [draw, black, fill={player2}, diamond] (3) {}
          child {node [draw, black, fill={player1}, circle] (4) {}
            child {node [draw, black, fill={terminal}, regular polygon, regular polygon sides=3, inner sep=0.08cm, label=below:{{\scriptsize[1, 3]}}] (5) {}
            edge from parent node[engine, sloped,below] {{\scriptsize $(\mathbf{\pi}_3, \T)$}}}
          edge from parent node[engine, sloped,below] {{\scriptsize $(\mathbf{\pi}_2, \T)$}}}
        edge from parent node[engine, sloped,below] {{\scriptsize $(\mathbf{\pi}_1, \T)$}}}
        child {node [draw, black, fill={player2}, diamond] (6) {}
          child {node [draw, black, fill={player1}, circle] (7) {}
            child {node [draw, black, fill={player2}, diamond] (8) {}
              child {node [draw, black, fill={terminal}, regular polygon, regular polygon sides=3, inner sep=0.08cm, label=below:{{\scriptsize[2, 0.5]}}] (9) {}
              edge from parent node[engine, sloped] {{\scriptsize $(\mathbf{\pi}_1, \T)$}}}
            edge from parent node[engine, sloped] {{\scriptsize $(\mathbf{\pi}_2, \F)$}}}
          edge from parent node[engine, sloped] {{\scriptsize $(\mathbf{\pi}_5, \T)$}}}
        edge from parent node[engine, sloped] {{\scriptsize $(\mathbf{\pi}_4, \F)$}}}
      edge from parent node[engine, sloped] {$H$}}
    edge from parent node[engine, sloped] {$H$}
    };

\end{tikzpicture}
\caption{ $\hat{G}$ after } \label{fig:scenario1_b}
\end{subfigure}
    \caption{Scenario 1
    }
    \label{fig:expand_scenario1}
\end{figure}

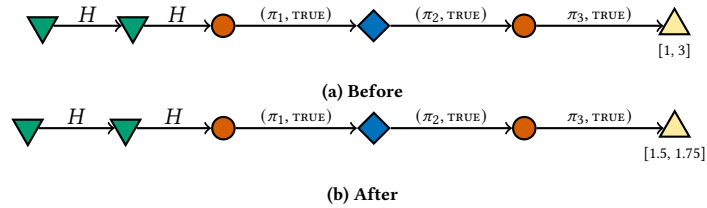
\begin{figure}[ht!]
\centering
    \begin{subfigure}[b]{\textwidth}
    \centering
\begin{tikzpicture}[thick,
    level 1/.style = {level distance = 12mm, sibling distance = 80mm},
    level 2/.style = {level distance = 12mm, sibling distance = 38mm},
    level 3/.style = {level distance = 20mm, sibling distance = 15mm},
    level 4/.style = {level distance = 20mm, sibling distance = 8mm},
    level 5/.style = {level distance = 20mm, sibling distance = 8mm},
    engine/.style = {inner sep = 1pt, above},
    grow=right]
    \node [draw, black, fill={chance}, regular polygon, regular polygon sides=3, rotate=180, inner sep=0.08cm] {}
    [black, ->]

    child {node [draw, black, fill={chance}, regular polygon, regular polygon sides=3, rotate=180, inner sep=0.08cm] (1) {} 
      child {node [draw, black, fill={player1}, circle] (2) {} 
        child {node [draw, black, fill={player2}, diamond] (3) {}
          child {node [draw, black, fill={player1}, circle] (4) {}
            child {node [draw, black, fill={terminal}, regular polygon, regular polygon sides=3, inner sep=0.08cm, label=below:{{\scriptsize[1, 3]}}] (5) {}
            edge from parent node[engine, sloped] {{\scriptsize $\mathbf{\pi}_3, \T)$}}}
          edge from parent node[engine, sloped] {{\scriptsize $(\mathbf{\pi}_2, \T)$}}}
        edge from parent node[engine, sloped] {{\scriptsize $(\mathbf{\pi}_1, \T)$}}}
      edge from parent node[engine, sloped] {$H$}}
    edge from parent node[engine, sloped] {$H$}
    };

\end{tikzpicture}
\caption{ Before} \label{fig:scenario2_a}
\end{subfigure} \\
\begin{subfigure}[b]{\textwidth}
\centering

\begin{tikzpicture}[thick,
    level 1/.style = {level distance = 13mm, sibling distance = 80mm},
    level 2/.style = {level distance = 13mm, sibling distance = 38mm},
    level 3/.style = {level distance = 20mm, sibling distance = 18mm},
    level 4/.style = {level distance = 20mm, sibling distance = 8mm},
    engine/.style = {inner sep = 1pt, above},
    grow=right]
    \node [draw, black, fill={chance}, regular polygon, regular polygon sides=3, rotate=180, inner sep=0.08cm] {}
    [black, ->]

    child {node [draw, black, fill={chance}, regular polygon, regular polygon sides=3, rotate=180, inner sep=0.08cm] (1) {} 
      child {node [draw, black, fill={player1}, circle] (2) {} 
        child {node [draw, black, fill={player2}, diamond] (3) {}
          child {node [draw, black, fill={player1}, circle] (4) {}
            child {node [draw, black, fill={terminal}, regular polygon, regular polygon sides=3, inner sep=0.08cm, label=below:{{\scriptsize[1.5, 1.75]}}] (5) {}
            edge from parent node[engine, sloped] {{\scriptsize $\mathbf{\pi}_3, \T)$}}}
          edge from parent node[engine, sloped] {{\scriptsize $(\mathbf{\pi}_2, \T)$}}}
        edge from parent node[engine, sloped] {{\scriptsize $(\mathbf{\pi}_1, \T)$}}}
      edge from parent node[engine, sloped] {$H$}}
    edge from parent node[engine, sloped] {$H$}
    };

\end{tikzpicture}
\caption{ After} \label{fig:scenario2_b}
\end{subfigure}
    \caption{Scenario 2}
    \label{fig:expand_scenario2}
\end{figure}

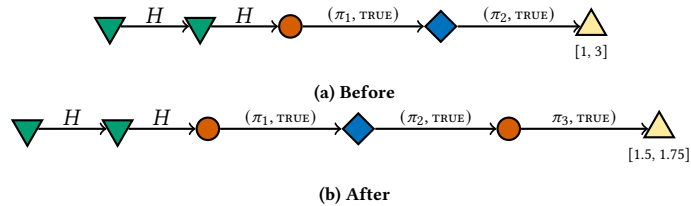
\begin{figure}[ht!]
\centering
    \begin{subfigure}[b]{\textwidth}
    \centering
\begin{tikzpicture}[thick,
    level 1/.style = {level distance = 12mm, sibling distance = 80mm},
    level 2/.style = {level distance = 12mm, sibling distance = 38mm},
    level 3/.style = {level distance = 20mm, sibling distance = 18mm},
    level 4/.style = {level distance = 20mm, sibling distance = 8mm},
    engine/.style = {inner sep = 1pt, above},
    grow=right]
    \node [draw, black, fill={chance}, regular polygon, regular polygon sides=3, rotate=180, inner sep=0.08cm] {}
    [black, ->]

    child {node [draw, black, fill={chance}, regular polygon, regular polygon sides=3, rotate=180, inner sep=0.08cm] (1) {} 
      child {node [draw, black, fill={player1}, circle] (2) {} 
        child {node [draw, black, fill={player2}, diamond] (3) {}
            child {node [draw, black, fill={terminal}, regular polygon, regular polygon sides=3, inner sep=0.08cm, label=below:{{\scriptsize[1, 3]}}] (5) {}
          edge from parent node[engine, sloped] {{\scriptsize $(\mathbf{\pi}_2, \T)$}}}
        edge from parent node[engine, sloped] {{\scriptsize $(\mathbf{\pi}_1, \T)$}}}
      edge from parent node[engine, sloped] {$H$}}
    edge from parent node[engine, sloped] {$H$}
    };

\end{tikzpicture}
\caption{ Before} \label{fig:scenario3_a}
\end{subfigure} \\
\begin{subfigure}[b]{\textwidth}
\centering

\begin{tikzpicture}[thick,
    level 1/.style = {level distance = 12mm, sibling distance = 80mm},
    level 2/.style = {level distance = 12mm, sibling distance = 38mm},
    level 3/.style = {level distance = 20mm, sibling distance = 18mm},
    level 4/.style = {level distance = 20mm, sibling distance = 8mm},
    level 5/.style = {level distance = 20mm, sibling distance = 8mm},
    engine/.style = {inner sep = 1pt, above},
    grow=right]
    \node [draw, black, fill={chance}, regular polygon, regular polygon sides=3, rotate=180, inner sep=0.08cm] {}
    [black, ->]

    child {node [draw, black, fill={chance}, regular polygon, regular polygon sides=3, rotate=180, inner sep=0.08cm] (1) {} 
      child {node [draw, black, fill={player1}, circle] (2) {} 
        child {node [draw, black, fill={player2}, diamond] (3) {}
          child {node [draw, black, fill={player1}, circle] (4) {}
            child {node [draw, black, fill={terminal}, regular polygon, regular polygon sides=3, inner sep=0.08cm, label=below:{{\scriptsize[1.5, 1.75]}}] (5) {}
            edge from parent node[engine, sloped] {{\scriptsize $\mathbf{\pi}_3, \T)$}}}
          edge from parent node[engine, sloped] {{\scriptsize $(\mathbf{\pi}_2, \T)$}}}
        edge from parent node[engine, sloped] {{\scriptsize $(\mathbf{\pi}_1, \T)$}}}
      edge from parent node[engine, sloped] {$H$}}
    edge from parent node[engine, sloped] {$H$}
    };

\end{tikzpicture}
\caption{After } \label{fig:scenario3_b}
\end{subfigure}
    \caption{Scenario 3}
    \label{fig:expand_scenario3}
\end{figure}

In the second scenario, the new history to be added to the tree in Figure~\ref{fig:scenario2_a} happens to be completely contained by an existing terminal node. Figure~\ref{fig:scenario2_b} reflects the resulting update where the set of terminal nodes in the tree remains unchanged, but the leaf utility at the existing node has been recalculated to be the average of the old sampled utility $[1, 3]$ and the new one $[2, 0.5]$. If the original happened to be an average of multiple samples with the same corresponding terminal history, the new utility would be just another sample. For instance, if the original terminal node's leaf utility $[1, 3]$ in Figure~\ref{fig:scenario2_a} was the result of two observations with sampled utilities $u^1 = [0, 2]$ and $u^2 = [2, 4]$, the node's utility in Figure~\ref{fig:scenario2_b} would be updated to the new average $[1.\overline{3}, 2.1\overline{6}]$.

In the third scenario depicted in Figure~\ref{fig:expand_scenario3}, the new history to be added to the tree in Figure~\ref{fig:scenario3_a} is identical to the one currently included in the tree, but with the addition of a second action from player 1. We see that in Figure~\ref{fig:scenario3_b}, the terminal node representing the conclusion of gameplay following player 2's action is now a decision node belonging to player 1, and the resulting leaf utility at the terminal node following player 1's action is a recalculated average of the original node's leaf utility $[1, 3]$ and the new terminal node's sampled reward $[2, 0.5]$, as described in the previous case.

\newpage
\section{GBI Pseudocode and Omitted Subroutines of \textsc{ComputeSPE}}\label{app:SPE}


\subsection{Subroutines for Organizing Subgames}

\begin{algorithm}[H]
\renewcommand{\thealgorithm}{}
\small
\floatname{algorithm}{GetSubgameRoots}
\caption{Algorithm that constructs a subtree $\Psi$ containing all the roots of the subgames within a game $G$}
\label{alg:get_roots}
\begin{algorithmic}
\Require{Input game $G$}

\State{Initialize $\Psi$ to contain only the game root $h_0$}

\For{$h \in D$, beginning with those closest to the root node}
\If{$|I(h)| == 1$ and if for every node $h'$ that succeeds $h$, every node $h'' \in I(h')$ succeeds $h$ (see $\varphi(h')$, $\varphi(h'')$)}
\State{Add $h$ to $\Psi$ via an outgoing edge from the closest node $h'$ in $\Psi$ that precedes $h$ (i.e. $(h', h)$)}
\EndIf
\EndFor

\vspace{1em}
\Return $\Psi$
\end{algorithmic}
\end{algorithm}

\begin{algorithm}[H]
\renewcommand{\thealgorithm}{}
\small
\floatname{algorithm}{GetSubgameGroups}
\caption{Algorithm that groups the subgames of $G$ by their height in $\Psi$}
\label{alg:get_pruning_seq}
\begin{algorithmic}
\Require{Input game $G$, subtree of subgame roots $\Psi$, maximum height $\ell$}

\State{Initialize each set $\Theta_k$ of subgames whose roots are in $\Psi$ at a height $k$ to $\emptyset$, for all $1 \leq k \leq \ell$}
\For{$h_{\theta} \in \Psi$}
\State{Add the identifier $\theta$ of subgame $G_{\theta}$ with root $h_{\theta}$ at height $k$ in the subtree $\Psi$ to $\Theta_k$}
\EndFor

\vspace{1em}
\Return $\{ \Theta_k \}_{k = 1}^{\ell}$
\end{algorithmic}
\end{algorithm}

\newpage
\subsection{GBI}

\begin{algorithm}[H]
\renewcommand{\thealgorithm}{}
\small
\floatname{algorithm}{GBI}
\caption{Generalized backward induction \citep{kaminski19}}
\label{alg:brute_compute_spe}
\begin{algorithmic}
\Require{Input game $G$}

\State{$\Psi \gets \textsc{GetSubgameRoots}(G)$}
\State{$\ell \gets$ height of $h_0$ in $\Psi$}
\State{$\{ \Theta_k \}_{k = 1}^{\ell} \gets \textsc{GetSubgameGroups}(G, \Psi, \ell)$}
\State{$\Sigma^{SPE}_1 \gets \{ \textsc{GetInitialSPE}(G, \Theta_1) \}$}
\For{$1 \leq k < \ell$}
\If{$\Sigma^{SPE}_k = \emptyset$}
\Return $\emptyset$
\EndIf
\State{$S^{k + 1} \gets \bigcup\limits_{\theta \in \Theta_{k + 1}, j \in N} \left( \bigtimes_{I \in \left.\cI_j\right|_{G_{\theta}}} \left.A_j\right|_{G_{\theta}}(I) \right)$}
\For{$\bsigma^k \in \Sigma^{SPE}_k$}
\For{$\bsigma^{k + 1} \in S^{k + 1}$}
\For{$\theta \in \Theta_{k + 1}$}
\If{ $\neg \textsc{IsNash}(\left.\bsigma^{k + 1}\right|_{G_{\theta}}, G_{\theta})$}
\State{Discard $\bsigma^k \cup \bsigma^{k + 1}$}
\EndIf
\State{Replace $G_{\theta}$ with terminal node and payoff $U\left( \left.\bsigma^{k + 1}\right|_{G_{\theta}} \right)$}
\EndFor
\State{Add $\bsigma^k \cup \bsigma^{k + 1}$ to $\Sigma^{SPE}_{k + 1}$}
\EndFor
\EndFor
\EndFor

\Return $\Sigma^{SPE}_{\ell}$
\end{algorithmic}
\end{algorithm}

\begin{algorithm}[H]
\renewcommand{\thealgorithm}{}
\small
\floatname{algorithm}{IsNash}
\caption{Algorithm for checking whether a partial solution $\left.\bsigma\right|_{G_{\theta}}$ is a NE of a subgame $G_{\theta}$ where $\theta \in \Theta_k$}
\label{alg:is_nash}
\begin{algorithmic}
\Require{Partial solution profile $\left.\bsigma\right|_{G_{\theta}}$, subgame $G_{\theta}$, $\theta \in \Theta_k$}
\vspace{1em}
\For{$j \in N$}
\State{$\left.S_j\right|_{G_{\theta}} = \times_{I \in \left.\cI_j\right|_{G_{\theta}}} \left.A_j\right|_{G_{\theta}}(I)$}
\For{$\bm{s}_j \in \left.S_j\right|_{G_{\theta}}$}
\If{$\hat{U}_j \left( \bm{s}_j, \left.\bsigma_{-j}\right|_{G_{\theta}} \right) > \hat{U}_j \left( \left.\bsigma\right|_{G_{\theta}} \right)$}
\Return False
\EndIf

\EndFor
\EndFor

\vspace{1em}
\Return True
\end{algorithmic}
\end{algorithm}

\newpage
\subsection{Subroutines for \textsc{ComputeSPE}}

\begin{algorithm}[H]
\renewcommand{\thealgorithm}{}
\small
\floatname{algorithm}{GetInitialSPE}
\caption{Algorithm that finds the initial set of partial SPE for the subgames closest to the terminal nodes in $G$}
\label{alg:get_spe_theta1}
\begin{algorithmic}
\Require{Input game $G$, subgame identifier set $\Theta_1$}

\State{Initialize $\bsigma^{SPE}_1 = \emptyset$}
\For{$\theta \in \Theta_1$}
\State{Find a NE for the subgame $\left.\bsigma^{NE}\right|_{G_{\theta}} = \textsc{SubgameCFR}(G_{\theta})$}
\State{$\bsigma^{SPE}_1 = \bsigma^{SPE}_1 \cup \left.\bsigma^{NE}\right|_{G_{\theta}}$}
\EndFor

\vspace{1em}
\Return{$\bsigma^{SPE}_1$}
\end{algorithmic}
\end{algorithm}

\begin{algorithm}[H]
\renewcommand{\thealgorithm}{}
\small
\floatname{algorithm}{\textsc{SubgameCFR}}
\caption{Adaptation of CFR for computing subgame SPE}
\label{alg:subgame_cfr}
\begin{algorithmic}
\Require{subgame $G_{\theta}$, partial SPE profile $\left.\bsigma^\mathit{SPE}\right|_{G_{\theta}}$, timesteps $T$}

\State{$\cI_{G_{\theta} \setminus \bsigma} \gets \{ I \in G_{\theta} \mid I \notin \left.\bsigma^\mathit{SPE}\right|_{G_{\theta}} \}$}
\If{$V(h_{\theta}) = 0$ and $\cI_{G_{\theta} \setminus \bsigma} = \emptyset$}
\Return $\left.\bsigma^\mathit{SPE}\right|_{G_{\theta}}$
\EndIf
\For{$I \in \cI_{G_{\theta} \setminus \bsigma}$}
\State $j \gets V(I)$
\State{$\bsigma^1(I)(a) \gets \frac{1}{\mid A_{j}(I) \mid}$ for all $a \in A_{j}(I)$}
\State{$R_I[a] \gets 0$ for for all $a \in A_{j}(I)$}
\State{$S_I[a] \gets 0$ for all $a \in A_{j}(I)$}
\EndFor

\For{$t \in \{1, \dotsc, T\}$}
\State{$u^{*} \gets \textsc{Traverse}\left( h_{G_{\theta}}, t, 1, 1, 1, G_{\theta}, \left.\bsigma^\mathit{SPE}\right|_{G_{\theta}} \right)$} 

\EndFor

\For{$I \in \cI_{G_{\theta} \setminus \bsigma}$}
\State{$\bsigma^{CFR}(I) \gets \textsc{Average}\left( \{ \bsigma^t(I) \}_{t = 1}^T \right)$}
\EndFor

\Return $\bsigma^{CFR} \cup \left.\bsigma^{SPE}\right|_{G_{\theta}}$
\end{algorithmic}
\end{algorithm}

\begin{algorithm}[H]
\renewcommand{\thealgorithm}{}
\small
\floatname{algorithm}{\textsc{Traverse}}
\caption{Completes a full CFR-traversal of subgame $G_{\theta}$}
\label{alg:traverse}
\begin{algorithmic}
\Require{Node $h$, time $t$, reach probabilities $r_0$, $r_1, r_2$; subgame $G_{\theta}$, partial SPE $\left.\bsigma^{SPE}\right|_{G_{\theta}}$}

\State{$j \gets V(h)$}
\If{$j = 0$}
\State{$u^{*} = \vec{0}$}
\For{ $x \in X(h)$}
\State{$u^{*} \gets u^{*} + P(x \mid h) \cdot \textsc{Traverse}\left( hx, \hspace{0.05em} t, \hspace{0.05em} r_1, \hspace{0.05em} r_2, \hspace{0.05em} r_0 \cdot P(x \mid h), \hspace{0.05em} G_{\theta}, \hspace{0.05em} \left.\bsigma^{SPE}\right|_{G_{\theta}} \right)$}
\EndFor
\Return $u^{*}$
\ElsIf{$h$ is terminal}
\Return $u(h)$
\ElsIf{$I(h) \in \left.\bsigma^{SPE}\right|_{G_{\theta}}$}
\Return $U(\left.\bsigma^{SPE}\right|_{G_{\theta}} \mid h)$
\EndIf

\State{$I \gets I(h)$}
\State{Initialize counterfactual utility $u^{*}_I \gets 0$}
\State{Initialize action utilities $u_I[a] \gets 0$ for $a \in A_{j}(I)$}
\For{$a \in A_{j}(I)$}
\If{$j = 1$}
\State{$u_I[a] \gets \textsc{Traverse}\left( ha, t, r_1 \cdot \bsigma^t(I)(a), r_2, r_0, G_{\theta}, \left.\bsigma^{SPE}\right|_{G_{\theta}} \right)$}
\Else
\State{$u_I[a] \gets \textsc{Traverse}\left( ha, t, r_1, r_2 \cdot \bsigma^t(I)(a), r_0, G_{\theta}, \left.\bsigma^{SPE}\right|_{G_{\theta}} \right)$}
\EndIf
\State{$u^{*}_I \gets u^{*}_I \mathrel{+} \bsigma^t(I)(a) \cdot u_I[a]$}
\EndFor

\State{$r_j, r_{-j} \gets \emptyset$}
\If{$j = 1$}
\State{$r_j \gets r_1$}
\State{$r_{-j} \gets r_2 r_0$}
\Else
\State{$r_j \gets r_2$}
\State{$r_{-j} \gets r_1 r_0$}
\EndIf
\For{$a \in A_{j}(I)$}
\State{$R_I(a) \gets R_I(a) \mathrel{+} r_{-j} \cdot \left(u_I[a][j] - u^{*}_I[j] \right)$}
\State{$S_I(a) \gets S_I(a) \mathrel{+} r_j \cdot \bsigma^t(I)(a)$}
\EndFor
\State{Update $\bsigma^{t + 1}(I)$ using $R_I(\cdot)$ values and regret-matching}

\Return $u^{*}_I$
\end{algorithmic}
\end{algorithm}

\newpage
\section{Omitted Proofs of Correctness and Runtime for \textsc{ComputeSPE}}\label{app:proofs}
Our analysis assumes the Nash solver runtime is represented by~$T$, which may or may not depend upon the size of the game tree.
Our algorithm combines the solver of choice with dynamic programming to solve for the SPE, working on each subgame (minus the information sets of the subgame already included in the partial SPE) for time $T$. If a solution is found for a subgame $G_{\theta_1}$, and the algorithm moves on to compute a solution for the subgame~$G_{\theta_2}$, where $G_{\theta_1} \subset G_{\theta_2}$, $G_{\theta_1}$ is not traversed again when finding the next solution.
Thus, if we assume that each part of the tree is traversed at most $T$ times, the runtime of the algorithm is also $O(T) \cdot O(\abs{H})$.
Alternatively, if we consider that in the worst-case, the number of actions in an information set is $A$ and the total number of information sets in $G$ is $\abs{\mathcal{I}}$, the runtime of the algorithm is $O \left(T A \cdot \abs{\mathcal{I}} \right)$, which is much tighter than that of GBI.

Our method for finding SPE applies to extensive-form games of imperfect information, as we demonstrate in the following lemma.

\begin{lemma}
\textsc{ComputeSPE} can find the SPE of any game $G$ of imperfect information.
\end{lemma}
\begin{proof}
There are two cases that arise when $G$ has imperfect information. In the first case, $G$ has no subgames besides itself. The subroutine \textsc{GetInitialSPE} within \textsc{ComputeSPE} is called on $G$ itself, as the height of the root $h_0$ in $\Psi$ must be 1. Since \textsc{GetInitialSPE} solves a given subgame using the black-box Nash solver, \textsc{ComputeSPE} returns the resulting NE, which is therefore the SPE.

In the second case, $G$ contains nontrivial subgames. \textsc{ComputeSPE} begins by first solving each of the subgames at height $k=1$ in the tree with the black-box Nash solver. The NE returned for each of these games must by definition be the SPE for each of these games. Consider this the base case for proof by induction. Then, the solution $\left.\bm{\sigma}^{SPE}\right|_{G_{\theta}}$ for any subgame $G_{\theta}$ at height $k$ is fixed, and the solver is applied to the subgame $G_{\theta'}$ at height $k + 1$ that contains it so as to find the optimal strategy $\left.\bm{\sigma}^{k+1}\right|_{G_{\theta'}}$ for that larger subgame (without overwriting the solutions for subgames at lower heights). $\left.\bm{\sigma}\right|_{G_{\theta'}}$ will consist of $\left.\bm{\sigma}^{SPE}\right|_{G_{\theta}}$ and the optimal joint strategy profile for the information sets that comprise the rest of $G_{\theta'}$ found via the Nash solver. This optimal profile is the SPE for that particular subgame. Since this continues for all subgames leading up to $h_0$, it follows by induction that the solution is ultimately the union of all SPE for the subgames of $G$, which by definition is the SPE. 
\end{proof}

\newpage
\section{Space and Runtime Requirements for Experiments}\label{app:great_lakes}
We give the runtime and memory requirements of TE-PSRO for different values of $M$, different choices of EVAL, and different choices of MSS. We also give the memory requirements of PSRO with a normal-form model. All values are for the trial that required the greatest amount of memory and for the trial that ran for the longest amount of time.

\begin{table}[h!]
    \centering
    \bgroup
    \def\arraystretch{1.1}
    \begin{tabular}{|c|c|c|c|c|}
    \hline
         $M$ & EVAL & MSS & Memory Used & Runtime in Hours \\
          \hline
        1 & NE & NE & 6.2GB & 10\\
        1 & NE & SPE & 6GB & 10\\
        1 & SPE & NE & 5.9GB & 9\\
        1 & SPE & SPE & 5GB & 8\\
        \hline
        2 & NE & NE & 12GB & 12\\
        2 & NE & SPE & 12GB & 16\\
        2 & SPE & NE & 11GB & 14\\
        2 & SPE & SPE & 12GB & 18\\
        \hline
        4 & NE & NE & 12GB & 11\\
        4 & NE & SPE & 10GB & 3\\
        4 & SPE & NE & 11GB & 12\\
        4 & SPE & SPE & 11GB & 7\\
        \hline
        8 & NE & NE & 10GB & 19\\
        8 & NE & SPE & 11GB & 6\\
        8 & SPE & NE & 11GB & 17 \\
        8 & SPE & SPE & 10GB & 8\\
        \hline
        16 & NE & NE & 11GB & 31\\
        16 & NE & SPE & 11GB & 30\\
        16 & SPE & NE & 11GB & 38\\
        16 & SPE & SPE & 11GB & 8\\
        \hline
        NF-PSRO & NE & NE & 20GB & 80\\
        \hline
    \end{tabular}
    \egroup
    \vspace{1em}
    \caption{Runtime and Memory Requirements of Experiments for Sequential Bargaining Game}
    \label{tab:great_lakes}
\end{table}

\newpage
\section{Omitted Experimental Results and Plots}\label{app:expts}
\subsection{Average Empirical Game Size over TE-PSRO Epochs}\label{app:game_size}

\begin{figure}[h]
    \centering
    \begin{subfigure}[b]{0.5\textwidth}
    \caption{\barg}
    \includegraphics[width=\textwidth]{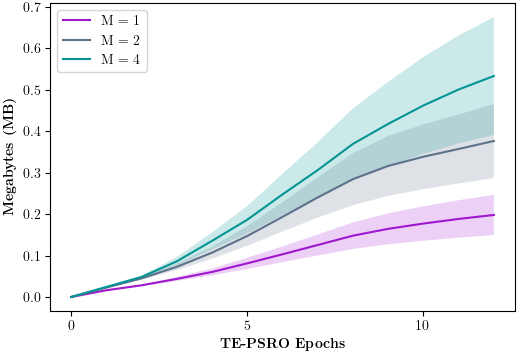}
    \label{fig:mem_barg}
    \end{subfigure}~
    \begin{subfigure}[b]{0.5\textwidth}
    \caption{\gengoofK{4}}
    \includegraphics[width=\textwidth]{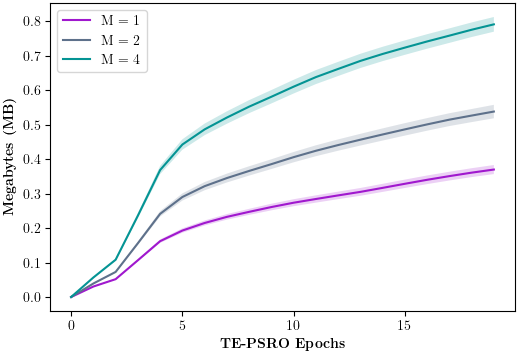}
     \label{fig:memory_gengoof}
    \end{subfigure}
    \caption{Memory in megabytes (MB) required by empirical game $\hat{G}$ over the course of TE-PSRO's runtime, averaged over all combinations of all parameters (except $M$) and seeds.}
    \label{fig:game_size_app}
\end{figure}
Fig.~\ref{fig:game_size} in \S\ref{sec:exp_results} and 
Fig.~\ref{fig:game_size_app} above together underscore the effectiveness of our heuristic of a choosing a subset of $M$ information sets to augment with outgoing edges based on the latest best response in keeping the empirical game tree growth tractable for both \barg and \gengoofK{4}. However, there is one stark difference between the results for these two games that merits elaboration: for \gengoofK{4}, the game size is significantly more consistent over time with a lower variance than for \barg as indicated by the smaller error bars for \gengoofK{4}. 

The likely reason for the above observation is that \gengoofK{4} by design consists of fewer rounds of player action ($3$ per player) than \barg ($5$ per player), which means that fewer epochs of TE-PSRO passed before a complete empirical history was included in the game of \gengoofK{4} instead of a default policy being assumed for the remainder of a history, as described in Section~\ref{sec:expand_game_tree}. For \barg, depending on the choice of $M$ and where the new best response policy was added, this results in higher variability in the lengths of these histories, hence higher variability in the number of information sets that comprised $\hat{G}$. Additionally, despite the size of $\hat{G}$ being more consistent, \gengoofK{4} did result in a larger empirical game tree than \barg (compare the vertical axes of Figs.~\ref{fig:mem_barg} and~\ref{fig:memory_gengoof}). This is not particularly surprising since \barg by design began with two stochastic rounds, each containing a binary event, while \gengoofK{4} included three stochastic rounds with a decreasing number of possible outcomes per event (four, then three, then two); this led to significantly more player information sets for \gengoofK{4} once these events were publically realized.

\subsection{Average Regret Over Time Given $M$ For Different Choices of MSS}

\begin{figure}[H]
    \centering
    \begin{subfigure}[b]{0.48\textwidth}
    \includegraphics[width=\textwidth]{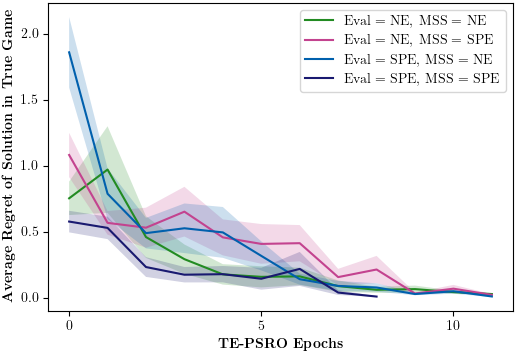}
    \caption{$M = 1$}
    \label{fig:app_regret_M1_no_NF}
    \end{subfigure}~
    \begin{subfigure}[b]{0.48\textwidth}
    \includegraphics[width=\textwidth]{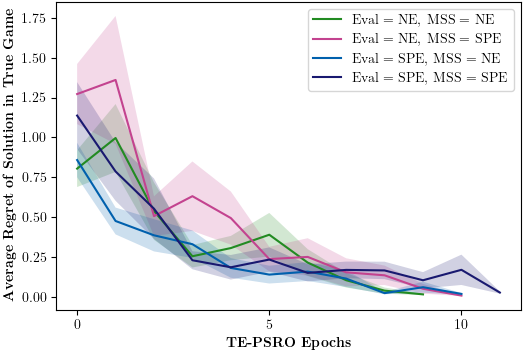}
    \caption{$M = 2$}
    \label{fig:app_regret_M2_no_NF}
    \end{subfigure}
    \begin{subfigure}[b]{0.48\textwidth}
    \includegraphics[width=\textwidth]{figs/true_regret_plots_by_num_br/true_regret_7_4_M4_no_NF.png}
    \caption{$M = 4$}
    \label{fig:app_regret_M4_no_NF}
    \end{subfigure}~
    \begin{subfigure}[b]{0.48\textwidth}
    \includegraphics[width=\textwidth]{figs/true_regret_plots_by_num_br/true_regret_7_3_M8_no_NF.png}
    \caption{$M = 8$}
    \label{fig:app_regret_M8_no_NF}
    \end{subfigure}
    \begin{subfigure}[b]{0.48\textwidth}
    \includegraphics[width=\textwidth]{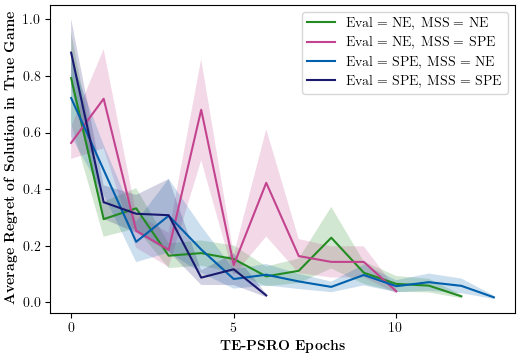}
    \caption{$M = 16$}
    \label{fig:app_regret_M16_no_NF}
    \end{subfigure}
    \caption{Average regret of evaluated solution from empirical game $\hat{G}$ in true game over the course of TE-PSRO's runtime, using different MSS's for exploration and different values of $M$.}
    \label{fig:app_true_regret_no_NF_by_M}
\end{figure}
\newpage

\subsection{Average Regret Over Time Given $M$ For Different Choices of MSS Compared to Normal-Form}

\begin{figure}[H]
    \centering
    \begin{subfigure}[b]{0.48\textwidth}
    \includegraphics[width=\textwidth]{figs/true_regret_plots_by_num_br/true_regret_6_16_M1.png}
    \caption{$M = 1$}
    \label{fig:app_regret_M1}
    \end{subfigure}~
    \begin{subfigure}[b]{0.48\textwidth}
    \includegraphics[width=\textwidth]{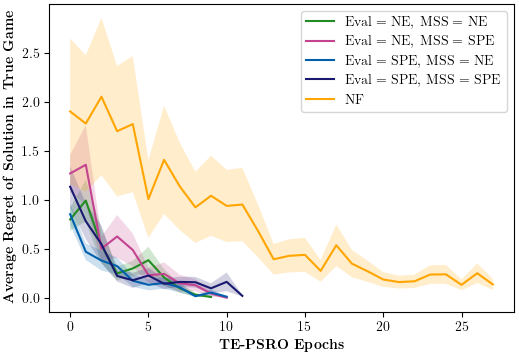}
    \caption{$M = 2$}
    \label{fig:app_regret_M2}
    \end{subfigure}
    \begin{subfigure}[b]{0.48\textwidth}
    \includegraphics[width=\textwidth]{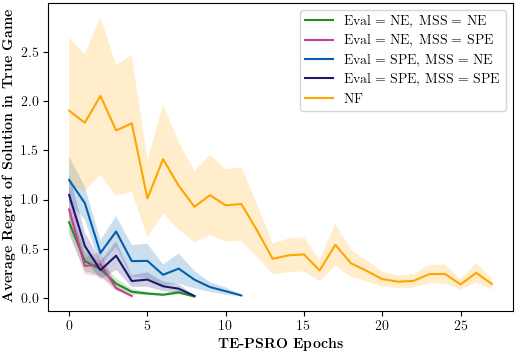}
    \caption{$M = 4$}
    \label{fig:app_regret_M4}
    \end{subfigure}~
    \begin{subfigure}[b]{0.48\textwidth}
    \includegraphics[width=\textwidth]{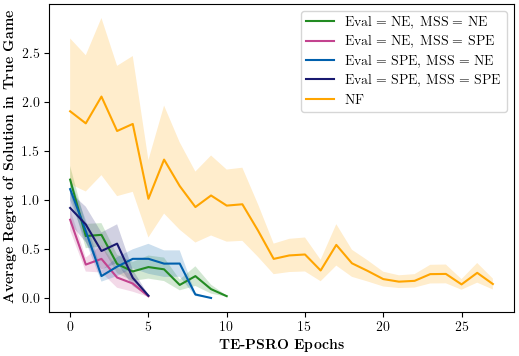}
    \caption{$M = 8$}
    \label{fig:app_regret_M8}
    \end{subfigure}
    \begin{subfigure}[b]{0.48\textwidth}
    \includegraphics[width=\textwidth]{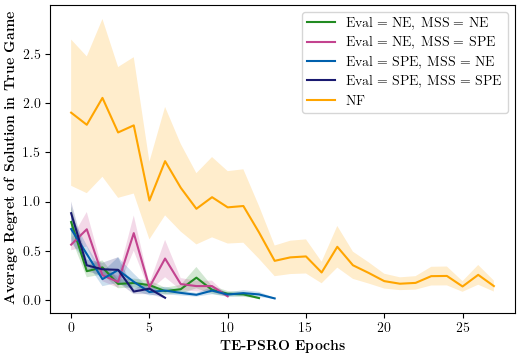}
    \caption{$M = 16$}
    \label{fig:app_regret_M16}
    \end{subfigure}
    \caption{Average regret of evaluated solution from empirical game $\hat{G}$ in true game over the course of TE-PSRO's runtime, using different MSS's for exploration and different values of $M$. Regret over time for PSRO using a normal-form model included as baseline.}
    \label{fig:app_true_regret}
\end{figure}
\newpage
\subsection{Average Regret Over Time Given MSS and Evaluation Strategy For Different Choices of $M$}

\begin{figure}[H]
    \centering
    \begin{subfigure}[b]{0.48\textwidth}
    \includegraphics[width=\textwidth]{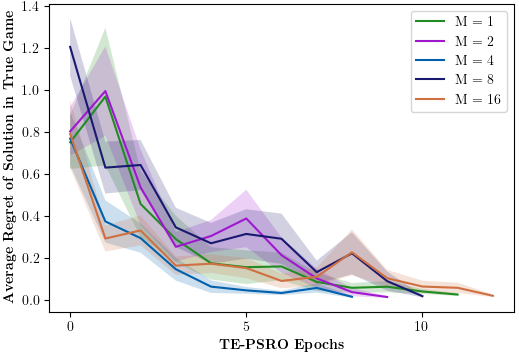}
    \caption{NE for true regret evaluation and for MSS}
    \label{fig:app_regret_ne_mss_ne_eval}
    \end{subfigure}~
    \begin{subfigure}[b]{0.48\textwidth}
    \includegraphics[width=\textwidth]{figs/true_regret_plots_by_mss/true_regret_spe_mss_ne_eval_7_3.png}
    \caption{NE for true regret evaluation, SPE for MSS}
    \label{fig:app_regret_spe_mss_ne_eval}
    \end{subfigure}
    \begin{subfigure}[b]{0.48\textwidth}
    \includegraphics[width=\textwidth]{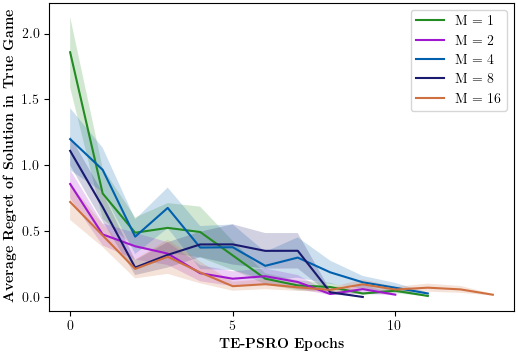}
    \caption{SPE for true regret evaluation, NE for MSS}
    \label{fig:app_regret_ne_mss_spe_eval}
    \end{subfigure}~
    \begin{subfigure}[b]{0.48\textwidth}
    \includegraphics[width=\textwidth]{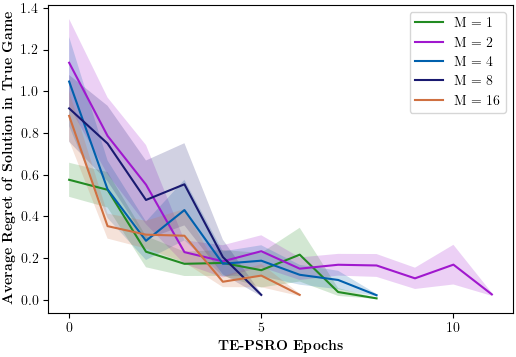}
    \caption{SPE for true regret evaluation and for MSS}
    \label{fig:app_regret_spe_mss_spe_eval}
    \end{subfigure}
    \caption{Average regret of evaluated solution from empirical game $\hat{G}$ in true game over the course of TE-PSRO's runtime, using different MSS's for exploration and different values of $M$.}
    \label{fig:app_true_regret_no_NF}
\end{figure}

\newpage
\subsection{Permutation Test Results for All $M$ and Choices of Evaluation Strategy}
\label{app:p-values}

\begin{table}[ht!]
    \centering
    \begin{tabular}{|c|c|c|c|}
    \hline
        $D_{\textsc{obs}}$ & $M$ & EVAL & p-value \\
    \hline
        0.210 & 1 & SPE & 0.286 \\
        -0.587 & 1 & NE & 0.862 \\
        -0.453 & 2 & SPE & 0.823 \\
        -0.152 & 2 & NE & 0.600 \\
        0.787 & 4 & SPE & 0.105 \\
        \textbf{0.124} & \textbf{4} & \textbf{NE} & \textbf{0.007} \\
        \textbf{0.938} & \textbf{8} & \textbf{SPE} & \textbf{0.021} \\
        \textbf{0.908} & \textbf{8} & \textbf{NE} & \textbf{0.006} \\
        0.488 & 16 & SPE & 0.086 \\
        -0.209 & 16 & NE & 0.636 \\
    \hline
    \end{tabular}
    \vspace{1em}
    \caption{Statistically significant p-values are highlighted in boldface. In these cases, we reject the null hypothesis and conclude that the using SPE as the MSS results in a statistically significant improvement in convergence to zero regret over using NE as the MSS, measured using the area under the curve. The area under the curve captures the number of TE-PSRO epochs until convergence as well as the distance on the y-axis from zero regret.}
    \label{tab:regret_perm_k5}
\end{table}

\subsection{Impact of $IR$-based coarsening of \gengoof empirical model on TE-PSRO performance.}
\label{app:IR_gengoof}

Recall that in Fig.~\ref{fig:abstract_spe_mss_chap8}, we see that $IR = [0]$ and $IR = [0, 1]$ yield the best performance regardless of MSS. This is a surprising result since intuition suggests that the more tree structure from the underlying game is incorporated into $\hat{G}$, the better TE-PSRO should perform, converging faster and to a tighter regret. 

A natural explanation for this observation is the complex interaction of $IR$ with the parameter $M$ that we imposed as a strict limit on the growth of $\hat{G}$ when adding best responses (Section~\ref{sec:best_response_augment}). The action edges of $\hat{G}$ either matched the actions of the underlying game in the rounds whose stochastic events were included or represented abstract DRL policies in the rounds of the game whose stochastic events were excluded. Additionally, the policies were learned as best responses for the \textit{entire underlying game} and then added to specific information sets as best responses, while in the rounds containing stochastic events, only the single actions at those information sets which yielded the best performance (i.e., had the maximum learned Q-value of the four available actions) were added as best responses. In this way, including the abstract policies without the stochastic event is a better option for exploiting the tree structure of the underlying game even though less tree structure is technically incorporated into $\hat{G}$. However, if we could add new best responses to all the information sets of $\hat{G}$ in each iteration of TE-PSRO without any limitation on the tree's growth, even for larger games, then it seems likely that including more stochastic rounds in $\hat{G}$ would be beneficial.

\section{Experimental Details for Best Response/DQN}\label{app:dqn}
The trained DQN consisted of a neural network with a single hidden layer containing 100 neurons and a fully connected ReLU activation. The action space was two-dimensional, as described in Section~\ref{sec:DOND}, but was flattened into a one-dimensional output vector of length $\vert \cA \vert$ for the network; it is important to note that Agent 1 was restricted from accepting a deal or walking during its first turn, since this would be illogical. The length of the input vector varied, depending on the item pool, valuation distribution, and outside offers that comprised the five different parameter settings of the bargaining game. The exploration policy was $\epsilon$-greedy and was initially set to $\epsilon = 1.0$ and decayed to the final, minimum $\epsilon$ value.
Each player was allotted the same number of training steps to learn the BR. During the DQN's experience replay, a batch size of 64 data points was sampled from the memory, and the Adam optimizer was used to update the network weights. The memory buffer was limited to 200k experiences. We did not allow the network to begin training and updating weights until after 10k timesteps had passed (i.e. at least 10k experiences were stored in the memory buffer). We considered the hyperparameters in Table~\ref{tab:DQN_considered_hyper} as candidates for tuning and found the hyperparameters in Table~\ref{tab:DQN_params} to perform best as the result of a randomized search over the grid space of values. After network training was complete, we used a temperature of 1.0 for the softmax function we use to select the $M$ information sets to which the best response policy label will be added.

\begin{table}[h!]
    \centering
    \begin{tabular}{|l||r|}
    \hline
    \multicolumn{2}{|c|}{Training Parameters} \\
    \hline
        Training Steps & 150000 \\
        Number of Hidden Neurons & 100 \\
        $\epsilon$ Decay & \textsc{Linear} \\
        Minimum $\epsilon$ & 0.02 \\
        Learning Rate $\alpha$ & $1\mathrm{e}{-4}$ \\
        Target Update Frequency & 2 \\
    \hline
    \end{tabular}
    \vspace{1em}
    \caption{Learned Hyperparameters of DQN for Sequential Bargaining Games.}
    \label{tab:DQN_params}
\end{table}

\begin{table}[h!]
    \centering
    \begin{tabular}{|l||r|}
    \hline
    \multicolumn{2}{|c|}{Hyperparameters} \\
    \hline
        Training Steps & $\{ 1\mathrm{e}{5}, 1.5\mathrm{e}{5}, 2\mathrm{e}{5}, 5\mathrm{e}{5}, 1\mathrm{e}{6} \}$ \\
        Number of Hidden Neurons & $\{ 50, 100, 200 \}$ \\
        $\epsilon$ Decay & $\{ \textsc{Linear}, \textsc{Exp} \}$ \\
        Minimum $\epsilon$ & $\{ 0.01, 0.02, 0.05 \}$ \\
        Learning Rate $\alpha$ & $\{ 1\mathrm{e}{-4}, 3\mathrm{e}{-4}, 5\mathrm{e}{-4}, 1\mathrm{e}{-3}, 3\mathrm{e}{-3} \}$ \\
        Target Update Frequency & $\{ 1, 2, 5 \}$ \\
    \hline
    \end{tabular}
    \vspace{1em}
    \caption{Hyperparameters of DQN Considered for Sequential Bargaining Games.}
    \label{tab:DQN_considered_hyper}
\end{table}


\end{document}